\documentclass[lettersize,journal]{IEEEtran}

\usepackage{amsmath,amsfonts}
\usepackage{array}
\usepackage[caption=false,font=footnotesize]{subfig}

\usepackage{textcomp}
\usepackage{stfloats}
\usepackage{url}
\usepackage{verbatim}
\usepackage{graphicx}
\def\BibTeX{{\rm B\kern-.05em{\sc i\kern-.025em b}\kern-.08em
    T\kern-.1667em\lower.7ex\hbox{E}\kern-.125emX}}
\usepackage{balance}

\input{symbols.sty}
\usepackage{float}
\usepackage{algorithm}
\usepackage{algpseudocode}

\usepackage[utf8]{inputenc}
 \usepackage{booktabs}

\usepackage{amsmath,graphicx}
\usepackage{amssymb}
\usepackage{xcolor}
\usepackage{mathrsfs}
\usepackage{hyperref}
\usepackage[nameinlink]{cleveref}
\usepackage{bm}
\usepackage{float}
\usepackage{enumitem}

\usepackage{hyperref}

\usepackage{ifthen}
\newboolean{if_comment} 
\setboolean{if_comment}{false} 
\usepackage{mathtools}
\newcommand{\commentg}[1]{%
  \ifthenelse{\boolean{if_comment}}{{\color{red} GL: #1}}{\ignorespaces}%
}

\newcommand{\commenty}[1]{%
  \ifthenelse{\boolean{if_comment}}{{\color{blue} YH: #1}}{\ignorespaces}%
}

\newcommand{\commentr}[1]{%
  \ifthenelse{\boolean{if_comment}}{{\color{green} Raj: #1}}{\ignorespaces}%
}

\newboolean{if_color} 
\setboolean{if_color}{false} 

\newcommand{\changes}[1]{%
  \ifthenelse{\boolean{if_color}}{\textcolor{green}{#1}}{#1}%
}

\begin{document}
\title{A Covariance Matching Approach to Graph Topology Identification}
\author{%
  Yongsheng~Han,~\IEEEmembership{Student Member,~IEEE,}%
  Raj~Thilak~Rajan,~\IEEEmembership{Senior Member,~IEEE,}%
  and~Geert~Leus,~\IEEEmembership{Fellow,~IEEE}%
  \thanks{This work is partially funded by the European Commission Key Digital Technologies Joint Undertaking - Research and Innovation (HORIZON-KDT-JU-2023-2-RIA), under grant agreement No 101139996, the ShapeFuture project - ''Shaping the Future of EU Electronic Components and Systems for Automotive Applications”. All authors are with Signal Processing Systems, Department of Microelectronics,
  Delft University of Technology, 2628~CD Delft, The~Netherlands
  (e-mail: \{y.han, r.t.rajan, g.j.t.leus\}@tudelft.nl). An open-source implementation is available at \url{https://github.com/asil-lab/GTI_Covariance_Match}.}%
}%
\markboth{Journal of \LaTeX\ Class Files,~Vol.~18, No.~9, September~2020}%
{How to Use the IEEEtran \LaTeX \ Templates}

\maketitle

\begin{abstract}
Graph topology identification (GTI) is a central challenge in networked systems, where the underlying structure is often hidden, yet nodal data are available. Conventional solutions to address these challenges rely on probabilistic models or complex optimization formulations, commonly suffering from non-convexity or requiring restrictive assumptions on acyclicity or positivity. In this paper, we propose a novel {covariance matching} (CovMatch) framework that directly aligns the empirical covariance of the observed data with the theoretical covariance implied by an underlying graph. We show that as long as the data-generating process permits an explicit covariance expression, CovMatch offers a unified route to topology inference.

We showcase our methodology on linear structural equation models (SEMs), showing that CovMatch naturally handles both undirected and general sparse directed graphs - whether acyclic or positively weighted - {without} explicit knowledge of these structural constraints. Through appropriate reparameterizations, CovMatch simplifies the graph learning problem to either a conic mixed integer program for undirected graphs or an orthogonal matrix optimization for directed graphs. Numerical results confirm that, even for relatively large graphs, our approach efficiently recovers the true topology and outperforms standard baselines in accuracy. These findings highlight CovMatch as a powerful alternative to log-determinant or Bayesian methods for GTI, paving the way for broader research on learning complex network topologies with minimal assumptions.
\end{abstract}


\section{Introduction}

\label{sec:intro}

Graph topology identification (GTI) remains a critical problem in graph signal processing (GSP), where systems are modeled as networks, yet their actual underlying structure is often invisible. Examples of such systems include brain functional connectivity networks and social networks. In these applications, while the direct graph structure is not observable, nodal data is typically available. For instance, in academic networks~\cite{liu2019shifu2}, the advisor-advisee links may not be visible, yet we can analyze collaborative patterns to uncover these connections. Similarly, in brain networks~\cite{mclntosh1994structural}, neural signals provide indirect clues about the connectivity. Therefore, the primary challenge in graph topology identification lies in deducing the hidden graph structure from these nodal observations, a task that is fundamental for analyzing and understanding the interactions within these networks.

A widely used approach to represent such hidden structures is the linear structural equation model (SEM)~\cite{sem_online,dyna_sem_giannkis,id_sem_circle,cai2013inference}. 
Note that although nonlinear versions of the SEM exist, we solely focus on the linear version in this paper, and will drop the term ``linear'' from now on. 
Under known exogenous inputs, it has been shown that, with diagonal input coupling and sufficiently rich exogenous signals, the underlying directed network can be uniquely recovered~\cite{id_sem_circle}.
However, when the exogenous variables remain unobserved, the SEM is often coupled with a directed acyclic graph (DAG) constraint, which is widely used for causal inference~\cite{pearl1998graphs,spirtes2000causation}.
Initially, it was shown that a score‑based search for the optimal DAG is NP‑hard~\cite{chickering2004large}. At that juncture, identifiability was an open problem, which hindered early progress in both theory and practice, thereby confining most efforts to small-scale graphs. 
This deadlock was first broken in the case of non-Gaussian disturbances by the so-called linear non-Gaussian acyclic model (LiNGAM), which established that the underlying DAG can be uniquely identified~\cite{shimizu2006linear}.
For the Gaussian case, this breakthrough was presented in \cite{id_sem_no_circle}, showing that DAG structures are indeed identifiable under certain conditions and further proposing a greedy search method to recover the underlying graph.

The above foundational works paved the way for subsequent developments, resulting in more efficient algorithms such as those in~\cite{DAG_notears,bello2022dagma,DAG_improved}, where the acyclicity constraint is replaced by powerful smooth characterizations that enable gradient-based optimization. These methods have achieved impressive performance, handling larger graphs with faster runtimes and lower estimation errors. 
Nevertheless, many real-world networks deviate from the DAG assumption, and learning more general directed graphs remains challenging. Existing approaches usually rely on a greedy search or additional metrics to guide optimization~\cite{amendola2020structure,yi2024filter}, leaving considerable space for further improvements.

In addition to directed graphs, undirected graphs have also been studied extensively in GSP and various methods have been proposed to address structure learning. For instance, approaches based on spectral templates (SpecTemp) leverage the eigen-decomposition of the data covariance matrix to estimate the underlying undirected topology~\cite{spectemp}. 
When the graph stationarity does not hold, a related line of work leverages richer second-order information—typically through multiple input--output diffusion pairs—to identify an unknown graph filter from covariance relations and then infer a sparse graph shift operator (GSO) via dedicated optimization, which has been developed for undirected graphs~\cite{shafipour2021identifying} and directed graphs via optimization on the Stiefel manifold~\cite{shafipour2018directed}.
These methods do not assume a specific graph filter or restrict its functional form, and thus often require a large number of samples.
Another covariance-based method is the covariance matching (CovMatch) framework, which was recently introduced in~\cite{ICASSP_covmatch}. It streamlines the estimation process by converting the graph estimation task into a binary variable inference problem. However, this framework only enforces the hollowness constraint (i.e., no self-loops) and does not use additional priors such as sparsity. Hence, accurate recovery typically requires very large sample sizes, even if the SEM structure is explicitly considered.
 In this work, we build on~\cite{ICASSP_covmatch} and generalize CovMatch to handle both undirected graphs and general directed graphs, with only a sparsity constraint on the number of active edges. Our proposed approach easily integrates additional constraints when needed and, even without explicit constraints (except for sparsity), can still recover diverse types of graphs, including DAGs and positively weighted directed graphs.

\subsection{Contributions}

In this work, we propose a novel framework for GTI based on CovMatch. The overarching idea is remarkably simple: as long as the data-generating model permits an explicit expression for the covariance, our approach provides a pathway toward estimating the underlying graph. The key contributions are as follows:

\paragraph{Unified covariance matching formulation}
We introduce a versatile framework that aligns the empirical covariance with the theoretical covariance derived from an underlying SEM. Although this paper focuses on an SEM, the proposed approach inherently accommodates any model whose covariance can be computed explicitly.

\paragraph{General sparse directed graphs}
To the best of our knowledge, our method is the first to handle {general sparse directed} graphs through a covariance matching paradigm for an SEM. It accurately infers network structures that go beyond the usual DAGs or undirected graphs.

\paragraph{Robustness without prior knowledge}
A notable feature is that our approach relies solely on the graph's inherent sparsity and requires no additional priors. No assumptions regarding acyclicity or positive edge weights are required, yet empirical results show strong performance across diverse network types.

\paragraph{Scalability to larger graphs}
We demonstrate that the proposed method can address relatively large graphs—even those that are not especially sparse—thanks to its binary or orthogonal variable reparameterization and efficient solution strategies. This scalability is particularly useful in real-world applications where network size is large.

Taken together, these contributions demonstrate that CovMatch can serve as a powerful method for GTI. By leveraging the generality of the covariance structure, our framework opens the door to new research directions in learning complex undirected and directed graphs.

\subsection{Notation}
\label{subsec:notation}

Scalars are written in italics, vectors in bold lowercase (e.g., $\mathbf a$), and matrices in bold uppercase (e.g., $\mathbf A$). All vectors are column vectors unless stated otherwise. Real domain is denoted by $\mathbb R$; $\mathbb R^{m\times n}$ is the set of real $m\times n$ matrices. For a matrix $\mathbf A$, the transpose is $\mathbf A^\top$, the inverse (when it exists) is $\mathbf A^{-1}$, and $\mathbf A^{-\top}=(\mathbf A^{-1})^\top$. The identity is $\mathbf I$; $\mathbf 0$ and $\mathbf 1$ denote all-zeros and all-ones arrays of compatible size. The trace and determinant are $\operatorname{tr}(\mathbf A)$ and $\det(\mathbf A)$. We write $\mathbf A\succeq 0$ for positive semidefinite and $\mathbf A\ge 0$ for elementwise nonnegativity.

For a vector $\mathbf a$, $\mathrm{diag}(\mathbf a)$ is the diagonal matrix with diagonal $\mathbf a$; for a matrix $\mathbf A$, $\mathrm{diag}(\mathbf A)$ extracts its diagonal into a vector. We also use $\operatorname{Diag}(\mathbf A)$ for the diagonal-matrix projection that keeps only the diagonal of $\mathbf A$ and zeros out off-diagonal entries, i.e., $\operatorname{Diag}(\mathbf A)\coloneqq \mathrm{diag}(\mathrm{diag}(\mathbf A))$. The Hadamard (elementwise) product is $\mathbf A\circ \mathbf B$. Powers of vectors are taken elementwise: for a scalar $\alpha$, $(\mathbf a)^\alpha$ raises each entry of $\mathbf a$ to $\alpha$. Matrix norms follow standard conventions: the Frobenius norm is $\|\mathbf A\|_F=\sqrt{\sum_{i,j}A_{ij}^2}$ and the elementwise $\ell_1$ norm is $\|\mathbf A\|_1=\sum_{i,j}|A_{ij}|$ (for a vector $\mathbf a$, $\|\mathbf a\|_1=\sum_i|a_i|$).

Matrix functions are used in the usual sense. The matrix exponential is $\exp(\mathbf A)=\sum_{k\ge 0}\mathbf A^k/k!$. The matrix logarithm $\log(\mathbf A)$ denotes the principal logarithm. When $\mathbf A\succ 0$, $\log(\mathbf A)$ is real. For orthogonal $\mathbf A$, $\log(\mathbf A)$ is real and skew-symmetric if $\det(\mathbf A)=+1$; if $\det(\mathbf A)=-1$, $\log(\mathbf A)$ is complex.

Binary vectors of length $N$ are written $\{-1,1\}^N$. The orthogonal group is
$\mathcal O(N)=\{\mathbf A\in\mathbb R^{N\times N}:\ \mathbf A^\top \mathbf A=\mathbf I\}.$
We use $\mathcal U(a,b)$ for the continuous uniform distribution on $[a,b]$ and $\mathcal U(\mathcal S)$ for the discrete uniform distribution over a finite set $\mathcal S$. A “uniform” draw on $\mathcal O(N)$ refers to the Haar measure on $\mathcal O(N)$ (left- and right-invariant). 

We use the Frobenius inner product $\langle \mathbf X,\mathbf Y\rangle=\operatorname{tr}(\mathbf X^\top \mathbf Y)$. Matrix differentials are denoted by $\mathrm d$; for a scalar function $f$ of a matrix $\mathbf A$ and a perturbation $\mathbf H$, the differential is $\mathrm d f(\mathbf A)[\mathbf H]=\langle \nabla_{\mathbf A} f(\mathbf A),\,\mathbf H\rangle$, which defines the Euclidean gradient $\nabla_{\mathbf A} f(\mathbf A)$ (also written $\partial f/\partial \mathbf A$).

\subsection{Graph Signal Processing}
\label{sec:preliminaries}

This subsection introduces the key concepts and notation from graph signal processing (GSP) that underpin our framework. We begin by defining the basic elements of a graph, its corresponding graph shift operator (GSO), and the notion of a graph signal.

\paragraph{Graphs}
Let $\mathcal{G} = (\mathcal{V}, \mathcal{E})$ denote a graph, where $\mathcal{V} = \{1, 2, \ldots, N\}$ is the set of nodes and $\mathcal{E} \subseteq \mathcal{V} \times \mathcal{V}$ is the set of edges. For a \emph{directed} edge from node $j$ to node $i$, we write $(j,i) \in \mathcal{E}$. In the \emph{undirected} case, whenever $j$ and $i$ are connected, both $(j,i)$ and $(i,j)$ belong to $\mathcal{E}$. We do not consider self-loops, so $(i,i) \notin \mathcal{E}$.  

\paragraph{Graph shift operator (GSO)}
A graph can be represented by a shift operator $\mathbf{S}\in\mathbb{R}^{N\times N}$, which shares the same support as $\mathcal{E}$. Specifically, $S_{ij}\neq 0$ only if $(j,i)\in \mathcal{E}$ or $i=j$. Common choices include the adjacency matrix or (normalized) Laplacian. In this work, we adopt the adjacency matrix, and since we exclude self-loops, $\mathrm{diag}(\mathbf{S})=\mathbf{0}$.

\paragraph{Graph signals}
A graph signal is formed by assigning a scalar value to each node, resulting in an $N\times 1$ vector $\mathbf{x} = [x_1, x_2, \dots, x_N]^\top$. The ordering of entries in $\mathbf{x}$ corresponds to the ordering of nodes in $\mathbf{S}$. GSP often assumes a relationship between the graph signal $\mathbf{x}$ and the GSO $\mathbf{S}$. While many such relationships exist, the next section introduces the SEM, which links $\mathbf{x}$ directly to $\mathbf{S}$ via a set of linear dependencies.

\section{Graph Topology Identification}
\label{sec:GTI_overview}

Graph topology identification (GTI) aims to recover the underlying matrix $\mathbf{S}$ (e.g., an adjacency matrix or more general graph shift operator) from data observations. Concretely, let $\{\mathbf{x}_t\}_{t=1}^T$ be a collection of random vectors in $\mathbb{R}^N$ drawn independently from a parametric distribution
$p(\mathbf{x}| \mathbf{S})$,
where $\mathbf{S}$ embodies the unknown network structure. Depending on the specific modeling assumptions, $\mathbf{S}$ may represent an undirected graph, a directed graph, or a more specialized structure (e.g., a DAG).

Characterising $p(\mathbf{x} | \mathbf{S})$ typically requires assumptions on the data-generating mechanism. In this paper, we focus on an SEM as an illustrative example for our study. Nonetheless, it is important to note that our approach has the potential to be applied to any model whose covariance structure can be computed explicitly.

\subsection{Structural Equation Model}
\label{sec:SEM_DAG_x}

In a structural equation model (SEM), each random vector $\mathbf{x} \in \mathbb{R}^N$ is generated according to the equation $\mathbf{x} = \mathbf{S}\,\mathbf{x} + \mathbf{e}$, where ${\bf S}$ represents the weighted adjacency matrix of the graph and $\mathbf{e}\in\mathbb{R}^N$ denotes unknown exogenous variables or, equivalently, a noise term. By rearranging the terms in the equation, one obtains $\mathbf{x} = (\mathbf{I} - \mathbf{S})^{-1}\,\mathbf{e}$. For convenience, we define $\mathbf{H} = (\mathbf{I} - \mathbf{S})^{-1}$, such that each entry of $\mathbf{x}$ can be interpreted as a weighted combination of the entries of $\mathbf{e}$, specifically $\mathbf{x} = \mathbf{H}\,\mathbf{e}$.

If $\mathbf{e}$ is zero mean with covariance $\boldsymbol{\Sigma}_{\mathbf{e}}$, then ${\bf x}$ is zero mean with covariance
\begin{equation}
\label{eq:cov_general}
\boldsymbol{\Sigma}_{\mathbf{x}}
= \mathbb{E}[\mathbf{x}\mathbf{x}^\top]
=(\mathbf{I}-\mathbf{S})^{-1}\boldsymbol{\Sigma}_{\mathbf{e}}(\mathbf{I}-\mathbf{S})^{-\top}
=\mathbf{H}\boldsymbol{\Sigma}_{\mathbf{e}}\mathbf{H}^\top.
\end{equation}

A particular instance arises when $\boldsymbol{\Sigma}_{\mathbf{e}}=\mathbf{I}$, which leads to $\boldsymbol{\Sigma}_{\mathbf{x}}=\mathbf{H}\mathbf{H}^{\top}$, simplifying to 
$\boldsymbol{\Sigma}_{\mathbf{x}}=\mathbf{H}^2=(\mathbf{I}-\mathbf{S})^{-2}$
when $\mathbf{S}$ is symmetric (undirected graph).

We construct $\mathbf{X}=[\mathbf{x}_1,\dots,\mathbf{x}_T]$ and $\mathbf{E}=[\mathbf{e}_1,\dots,\mathbf{e}_T]$, where $T$ is the number of samples. The matrix form of the SEM then becomes
\begin{equation}
\label{eq:SEM_matrix}
\mathbf{X}=\mathbf{S}\mathbf{X}+\mathbf{E}.
\end{equation}

\subsection{Maximum Likelihood Estimation and Its Limitations}
\label{subsec:Gaussian_SEM_likelihood}

In the special case where the exogenous variables are Gaussian, i.e., the noise is Gaussian, we obtain $p({\bf x}| {\bf S}) = \mathcal{N}({\bf 0}, \boldsymbol{\Sigma}_{\mathbf{x}})$, where $\boldsymbol{\Sigma}_{\mathbf{x}}$ is parameterized by ${\bf S}$ as in~\eqref{eq:cov_general}.
In other words, we obtain
\begin{equation}
p(\mathbf{x} | {\mathbf{S}})
=
\frac{1}{(2\pi)^{N/2}\bigl| {\boldsymbol{\Sigma}}_{\mathbf{x}}\bigr|^{1/2}}
\exp\!\Bigl(
-\tfrac{1}{2}\mathbf{x}^\top {\boldsymbol{\Sigma}}_{\mathbf{x}}^{-1}\mathbf{x}
\Bigr).
\end{equation}
The negative log-likelihood (up to additive constants) for the set of vectors $\{\mathbf{x}_t\}_{t=1}^T$ (or the matrix ${\bf X}$) thus yields the classical objective function~\cite{connect_dots}
\begin{equation}
\label{eq:ML_objective}
f(\hat{\bf S}) =\log\det\bigl(\hat{\boldsymbol{\Sigma}}_{\mathbf{x}}\bigr)
+
\mathrm{tr}(
\hat{\boldsymbol{\Sigma}}_{\mathbf{x}}^{-1}\mathbf{C}_{\mathbf{x}} ),
\quad
\end{equation}
where
$
\mathbf{C}_{\mathbf{x}}
=
\frac{1}{T}\mathbf{X}\mathbf{X}^\top$
denotes the sample covariance matrix. Note that here we use a hat on $\hat{\bf S}$ and $\hat{\boldsymbol \Sigma}_{\bf x}$ to indicate that they are optimization variables. This is in order to distinguish them from the ground truth ${\bf S}$ and ${\boldsymbol \Sigma}_{\bf x}$. The same will hold for other symbols throughout this manuscript.

Direct optimization of~\eqref{eq:ML_objective} is difficult, as the log-determinant term introduces strong non-convexity. A common workaround is to drop this component, leading to a {signal matching} (SigMatch) objective $\|\mathbf{X}-\hat{\mathbf{S}}\mathbf{X}\|_{F}^{2}$ \cite{aragam2015learning}. 
Such treatment is widely adopted in DAG settings, where $ \log\!\det (\hat{\boldsymbol{\Sigma}}_\mathbf{x})$ is equal to  $ \log\!\det ({\boldsymbol{\Sigma}}_\mathbf{e})$ and thus not related to $\hat{\mathbf{S}}$~\cite{id_sem_no_circle}. When $\boldsymbol{\Sigma}_{\mathbf{e}}$ is known, the log-determinant term reduces to a constant and can consequently be omitted \cite{DAG_notears}. 
For general graphs, however, eliminating the log–determinant term does not guarantee convergence to the correct solution, even in the undirected case (see Appendix~\ref{appendix:naive_proof}).

In this work, our aim is to develop strategies that bypass the complexities of optimizing the log-determinant term, without resorting to overly simplistic approximations. This preserves the potential to recover the correct network structure. In the next two sections, we present our covariance matching (CovMatch) framework for achieving this balance and derive a tractable optimization problem that can handle both undirected and directed graphs.

\section{Covariance Matching}

Instead of directly working with the likelihood, we adopt a covariance matching (CovMatch) approach. The motivation stems from the fact that, for a zero-mean Gaussian distribution, the covariance matrix fully characterizes the underlying distribution. Consequently, a natural idea is to align the sample covariance $\mathbf{C}_{\mathbf{x}}=\tfrac{1}{T}\mathbf{X}\mathbf{X}^\top$ with the theoretical covariance $\boldsymbol{\Sigma}_{\mathbf{x}}$. Recalling from \eqref{eq:cov_general} that
\begin{equation}
\boldsymbol{\Sigma}_{\mathbf{x}}
=
\mathbf{H}\boldsymbol{\Sigma}_{\mathbf{e}}\mathbf{H}^\top,
\quad
\mathbf{H}
=
(\mathbf{I}-\mathbf{S})^{-1},
\end{equation}
we aim to solve
\begin{equation}
\label{eq:general_covmatch}
\begin{aligned}
\mathbf{H}^*
=
\arg\min_{\hat{\mathbf{H}} \in \mathscr{H}}
\|\hat{\mathbf{H}}\boldsymbol{\Sigma}_{\mathbf{e}}\hat{\mathbf{H}}^\top 
- 
\mathbf{C}_{\mathbf{x}} \|_{F}^{2},
\end{aligned}
\end{equation}
where $\mathscr{H}$ denotes the feasible set of $\hat{\mathbf{H}}$ implied by the relevant graph structural assumptions (e.g., prohibiting self-loops, enforcing sparsity, requiring a DAG, guaranteeing positive weights, or imposing symmetry). By focusing on second-order statistics directly, this formulation provides a flexible alternative to log-determinant-based methods and can readily incorporate additional constraints as dictated by the application context.

Although the above objective remains non-trivial to solve in practice, it is worth noting that in the absence of any structural constraints, the best we can do is to estimate $\hat{\mathbf{H}}$ such that it exactly reconstructs $\mathbf{C}_{\mathbf{x}}$. Moreover, as $T \to \infty$, if the correct $\mathbf{H}$ (or equivalently $\mathbf{S}$) is identified, the objective converges to zero. 
This observation motivates a different perspective in which we reverse the roles of the objective and the constraints:
\begin{equation}
\label{eq:role_reversal}
\renewcommand{\arraystretch}{1.15}
\begin{array}{c@{\hspace{0.4em}}c@{\hspace{0.4em}}c}
\text{Classical} & & \text{CovMatch} \\[0.4em]
\begin{aligned}
\min_{\hat{\mathbf H}}\;\; & \big\|\hat{\mathbf H}\boldsymbol{\Sigma}_{\mathbf e}\hat{\mathbf H}^{\top}-\mathbf C_{\mathbf x}\big\|_F^{2} \\
\text{s.t.}\;\; & \hat{\mathbf H}\in\mathscr H
\end{aligned}
&
\ \rightarrow \ 
&
\begin{aligned}
\min_{\hat{\mathbf H}}\;\; & r_{\mathscr H}(\hat{\bf H}) \\
\text{s.t.}\;\; & \hat{\mathbf H}\boldsymbol{\Sigma}_{\mathbf e}\hat{\mathbf H}^{\top}=\mathbf C_{\mathbf x}
\end{aligned}
\end{array}
,
\end{equation}
where $r_{\mathscr H}(\hat{\bf H})$ represents a regularization term related to the constraint $\hat{\mathbf H}\in\mathscr H$.
In this formulation, instead of directly optimizing $\hat{\mathbf S}$ to minimize the covariance fitting term, we collect all matrices $\hat{\mathbf S}$ that exactly fit the sample covariance matrix and then select the one that best satisfies our desired constraints. Note that applying that same reversal of objective and constraints in the maximum likelihood formulation leads to an equivalent outcome (see Appendix~\ref{appendix:MLE}).

Inspired by the CovMatch perspective, we next examine what can be inferred from the equality
\begin{equation}
\hat{\mathbf{H}}\boldsymbol{\Sigma}_{\mathbf{e}}\hat{\mathbf{H}}^\top 
= 
\mathbf{C}_{\mathbf{x}},
\end{equation}
which holds when the objective is forced to be zero. We show how to map this condition into a more structured, lower-dimensional variable estimation problem, thereby simplifying the subsequent optimization. Specifically, this section handles the covariance matching objective for the undirected and directed cases separately, while deferring structural constraints to the next section. For clarity, we focus on a special case, i.e., $\boldsymbol{\Sigma}_{\mathbf{e}} = \mathbf{I}$, and relegate the discussion on the general scenario to Appendix~\ref{appendix:covmatch_general}.

\subsection{Undirected Graphs via Eigenvalue Decomposition}
\label{sec:undirected_case}

When the graph is undirected and the noise covariance is $\boldsymbol{\Sigma}_{\mathbf{e}} = \mathbf{I}$, our goal is to enforce $\hat{\mathbf{H}}^2 = \mathbf{C}_{\mathbf{x}}$. Since $\hat{\mathbf{H}}$ is symmetric, an eigenvalue decomposition (EVD) is the natural choice of decomposition. Hence, rather than estimating $\hat{\mathbf{H}}$ directly, we parameterize it via its eigenvalues and eigenvectors. Specifically, assuming $\hat{\mathbf{H}} = \hat{\mathbf{U}}\,
\mathrm{diag}(\hat{\boldsymbol{\lambda}})\,
\hat{\mathbf{U}}^\top$, we obtain $\hat{\mathbf{H}}^2 = \hat{\mathbf{U}}\,
\mathrm{diag}(\hat{\boldsymbol{\lambda}}^2)\,
\hat{\mathbf{U}}^\top$. 
Similarly, let 
$\mathbf{C}_{\mathbf{x}} = \mathbf{U}_{\mathbf{x}}\,
\mathrm{diag}(\boldsymbol{\lambda}_{\mathbf{x}})\,\mathbf{U}_{\mathbf{x}}^\top$
be its EVD. Matching $\hat{\mathbf{H}}^2$ to $\mathbf{C}_{\mathbf{x}}$ implies that we can set
$\hat{\mathbf{U}}=\mathbf{U}_{\mathbf{x}}$
and
$\hat{\boldsymbol{\lambda}}^2=\boldsymbol{\lambda}_{\mathbf{x}}$.
However, enforcing $\hat{\boldsymbol{\lambda}}^2 = \boldsymbol{\lambda}_{\mathbf{x}}$ introduces a sign ambiguity for each eigenvalue. To resolve it, we define a binary vector $\hat{\mathbf{q}}\in\{-1,1\}^N$  as an estimate of $\operatorname{sign}(\boldsymbol\lambda)$ and rewrite
$\,\hat{\boldsymbol{\lambda}} = \operatorname{diag}(\hat{\mathbf{q}})\,\boldsymbol{\lambda}_{\mathbf{x}}^{1/2}$, where $(\cdot)^{1/2}$ denotes the positive square root.
As a result, we obtain
\begin{equation}
\hat{\mathbf{H}}
=
\mathbf{U}_{\mathbf{x}}
\,\mathrm{diag}(\hat{\mathbf{q}})
\,\mathrm{diag}\bigl(\boldsymbol{\lambda}_{\mathbf{x}}^{1/2}\bigr)
\,\mathbf{U}_{\mathbf{x}}^\top,
\end{equation}
and substituting $\hat{\mathbf{H}}$ into 
$\hat{\mathbf{S}} = \mathbf{I}-\hat{\mathbf{H}}^{-1}$
yields
\begin{equation}
\label{eq:S_hat_undir}
\hat{\mathbf{S}}
=
\mathbf{I}
-
\mathbf{U}_{\mathbf{x}}
\,\mathrm{diag}(\hat{\mathbf{q}})
\,\mathrm{diag}\bigl(\boldsymbol{\lambda}_{\mathbf{x}}^{-1/2}\bigr)
\,\mathbf{U}_{\mathbf{x}}^\top,
\end{equation}
where we make use of the fact that $\hat{\mathbf{q}}^{-1} = \hat{\mathbf{q}}$.
Ultimately, the estimation task reduces to determining $\hat{\mathbf{q}}$, which captures the sign ambiguity of each eigenvalue. Note that when $T \rightarrow \infty$, we can only 
 identify a valid ground truth ${\bf U}$
(i.e., the eigenvectors of the ground truth ${\bf H}$) when ${\bf H}$ has no two non-zero eigenvalues that are negatives of each other, i.e., 
$\nexists\, i\neq j: \ \lambda_i = -\lambda_j$
(see Appendix~\ref{appendix:proof_of_identifiability}).

\subsection{Directed Graphs via Singular Value Decomposition}

A similar method is applied in the {directed} graph scenario, where
$\hat{\mathbf{H}} = (\mathbf{I} - \hat{\mathbf{S}})^{-1}$
may be non-symmetric, yet we still enforce
$\hat{\mathbf{H}}\,\hat{\mathbf{H}}^\top = \mathbf{C}_{\mathbf{x}}$.
Now, let 
$\,\hat{\mathbf{H}} = \hat{\mathbf{U}}\,\mathrm{diag}(\hat{\boldsymbol{\lambda}})\,\hat{\mathbf{V}}^\top$
be its singular value decomposition (SVD), and let
$\,\mathbf{C}_{\mathbf{x}} = \mathbf{U}_{\mathbf{x}}\,\mathrm{diag}(\boldsymbol{\lambda}_{\mathbf{x}})\,\mathbf{U}_{\mathbf{x}}^\top$
be again its EVD. Then matching
$\hat{\mathbf{H}}\,\hat{\mathbf{H}}^\top = \hat{\mathbf{U}}\,\mathrm{diag}(\hat{\boldsymbol{\lambda}}^2)\,\hat{\mathbf{U}}^\top$ to
$\mathbf{U}_{\mathbf{x}}\,\mathrm{diag}(\boldsymbol{\lambda}_{\mathbf{x}})\,\mathbf{U}_{\mathbf{x}}^\top$ means we can select $\hat{\mathbf{U}} = \mathbf{U}_{\mathbf{x}}$
and
$\,\hat{\boldsymbol{\lambda}} = \boldsymbol{\lambda}_{\mathbf{x}}^{1/2}$.
We now obtain
\begin{equation}
\hat{\mathbf{H}}
=
\mathbf{U}_{\mathbf{x}}
\,\mathrm{diag}\bigl(\boldsymbol{\lambda}_{\mathbf{x}}^{1/2}\bigr)
\,\hat{\mathbf{V}}^\top,
\end{equation}
and substituting $\hat{\bf H}$ into
$\,\hat{\mathbf{S}} = \mathbf{I}-\hat{\mathbf{H}}^{-1}$
yields
\begin{equation}
\label{eq:S_hat_dir}
\hat{\mathbf{S}}
=
\mathbf{I}
-
\hat{\mathbf{V}}
\,\mathrm{diag}\bigl(\boldsymbol{\lambda}_{\mathbf{x}}^{-1/2}\bigr)
\,\mathbf{U}_{\mathbf{x}}^\top.
\end{equation}
Unlike the undirected case, singular values are nonnegative by design, so there is no separate sign ambiguity. Instead, any remaining freedom is subsumed in the orthogonal matrix factor $\hat{\mathbf{V}}$, creating a rotation ambiguity. On the one hand, this makes the directed scenario more challenging, but on the other hand, when $T \rightarrow \infty$, 
we can now always identify a valid ground truth ${\bf U}$ 
(i.e., the left singular vectors of the ground truth ${\bf H}$) and hence no further assumptions are required as in the undirected case (see Appendix \ref{appendix:proof_of_identifiability_2}).

\section{Incorporating Structural Constraints}
\label{sec:constraints}

Thus far, we have focused on matching the model covariance to the sample covariance without imposing explicit structural requirements. 
In practice, however, domain knowledge often dictates additional constraints: for example, no self-loops in neuronal connectivity graphs~\cite{rubinov2010complex}, sparse interactions in large-scale networks, where each node only connects to a limited subset of others~\cite{BarabasiNS}, directed acyclicity in causal inference models~\cite{pearl1998graphs}, or positive edge weights~\cite{spectemp}.  
In this section, we illustrate two representative constraints, each of which can be implemented via an appropriate penalty function $h(\hat{\mathbf{S}})$.

\subsection{No Self-Loops}
A common requirement in certain graph structures is the absence of self-loops, i.e., the diagonal entries of the adjacency matrix $\hat{\mathbf{S}}$ should be zero. 
Translating this hard constraint into an objective-based formulation, we can introduce the penalty term
\begin{equation}
\label{eq:diag}
h_{\mathrm{diag}}(\hat{\mathbf{S}})
=
\|\mathrm{diag}(\hat{\mathbf{S}})\|_{2}^{2},
\end{equation}
thereby encouraging the diagonal to be small without strictly forcing it to be zero. 

For the undirected graph problem, a diagonal constraint is often sufficient. Intuitively, we only need to resolve $N$ binary variables $\hat{ q}_i$ in  $\hat{\mathbf{S}}$ as defined in~\eqref{eq:S_hat_undir}, and enforcing $\mathrm{diag}(\hat{\mathbf{S}})=\mathbf{0}$ supplies $N$ equations. In Appendix~\ref{appendix:proof_of_identifiability}, we provide a rigorous proof of when and why this condition leads to a successful recovery when $T \to \infty$.

In the directed case, however, the impact of the diagonal constraint diminishes as the size of the graph increases. This is because the relevant orthogonal matrix $\hat{\bf V}$ in $\hat{\mathbf{H}}$ has $\frac{N(N-1)}{2}$ degrees of freedom, far exceeding the $N$ equations imposed by $\mathrm{diag}(\hat{\mathbf{S}}) = {\bf 0}$.

\subsection{Sparsity}
Many real-world graphs exhibit a limited number of connections between nodes. To promote sparsity, we simply introduce an $\ell_{1}$-type penalty on $\hat{\mathbf{S}}$:
\begin{equation}
\label{eq:sparse}
h_{\mathrm{sparse}}(\hat{\mathbf{S}})
=
\|\hat{\mathbf{S}}\|_{1},
\end{equation}
where $\|\cdot\|_{1}$ denotes the elementwise absolute sum. This penalty, similar in spirit to the Lasso in sparse regression \cite{friedman2008sparse}, shrinks small coefficients toward zero while retaining larger ones, thus favoring a solution with fewer nonzero entries.

Beyond a global $\ell_1$ penalty, one may also add a term that limits the number of active edges per node. For instance, augmenting the objective with a function of the form $\max_i \| [\hat{\mathbf{S}}]_{i,:}\|_{1}$ (or $\max_j \| [\hat{\mathbf{S}}]_{:,j}\|_{1}$) bounds the total connection strength emanating from (or incoming to) each node, thereby controlling the maximum number of nonzero edges per row or column in a relaxed sense.

\subsection{Combining Constraints with CovMatch}

Combining the structural constraints~\eqref{eq:diag} and~\eqref{eq:sparse} leads to a purely structural objective of the form

\begin{equation}
\label{eq:structural}
h(\hat{\mathbf{S}})
=
h_{\mathrm{diag}}(\hat{\mathbf{S}})
+
\alpha h_{\mathrm{sparse}}(\hat{\mathbf{S}})
,
\end{equation}
where $\alpha\ge0$ adjusts the relative importance of sparsity over hollowness (no self-loops). In fact, additional constraints can readily be incorporated within the same framework, such as enforcing positive edge weights, bounding node degrees, or imposing a DAG structure. However, even without these extra constraints, our method already performs very well; in the experiments, we illustrate this point using the DAG case as an example in Section~\ref{subsec:exp_dag}.

For an undirected graph,
combining $h(\hat{\mathbf{S}})$ with the covariance matching constraint~\eqref{eq:S_hat_undir}, we obtain the following optimization problem:
\begin{equation}
\label{eq:undirected_structural_opt}
\begin{aligned}
\min_{\hat{\mathbf{q}}  \in  \{-1,1\}^N}
&\quad
h\bigl(\hat{\mathbf{S}}\bigr)
\\
\text{subject to}
&\quad
\hat{\mathbf{S}}
=
\mathbf{I}
-
\mathbf{U}_{\mathbf{x}} 
\mathrm{diag}(\hat{\mathbf{q}}) 
\mathrm{diag}\!\bigl(\boldsymbol{\lambda}_{\mathbf{x}}^{-1/2}\bigr)
 \mathbf{U}_{\mathbf{x}}^\top.
\end{aligned}
\end{equation}

Similarly, for a directed graph, combining $h(\hat{\mathbf{S}})$ with the corresponding covariance matching constraint~\eqref{eq:S_hat_dir} yields
\begin{equation}
\label{eq:directed_structural_opt}
\begin{aligned}
\min_{\hat{\mathbf{V}}: \hat{\mathbf{V}}^\top\hat{\mathbf{V}}=\mathbf{I}}
&\quad
h\bigl(\hat{\mathbf{S}}\bigr)
\\
\text{subject to}
&\quad
\hat{\mathbf{S}}
=
\mathbf{I}
-
\hat{\mathbf{V}} 
\mathrm{diag}\bigl(\boldsymbol{\lambda}_{\mathbf{x}}^{-1/2}\bigr) 
\mathbf{U}_{\mathbf{x}}^\top.
\end{aligned}
\end{equation}

In the next sections, we will study these two problems in more detail.

\section{Covariance Matching for Undirected SEM}
\label{sec:undirected_solution}

We first focus on the {undirected} graph setting. 
Combining the purely structural objective $h(\hat{\mathbf{S}})$ from~\eqref{eq:structural} with the covariance matching constraint from \eqref{eq:S_hat_undir} yields
\begin{equation}
\label{eq:undirected_structural_opt_2}
\begin{aligned}
\min_{\hat{\mathbf{q}}  \in  \{-1,1\}^N}
&\quad
 h_{\mathrm{diag}}(\hat{\mathbf{S}})
+
\alpha h_{\mathrm{sparse}}(\hat{\mathbf{S}})
\\
\text{subject to}
&\quad
\hat{\mathbf{S}}
=
\mathbf{I}
-
\mathbf{U}_{\mathbf{x}} 
\mathrm{diag}(\hat{\mathbf{q}}) 
\mathrm{diag}\!\bigl(\boldsymbol{\lambda}_{\mathbf{x}}^{-1/2}\bigr)
 \mathbf{U}_{\mathbf{x}}^\top.
\end{aligned}
\end{equation}
Here, the only additional constraints we consider are the hollowness and sparsity requirements.

\subsection{Hollowness Constraint}
Substituting the expression for $\hat{\bf S}$ from~\eqref{eq:undirected_structural_opt_2} into $h_{\mathrm{diag}}(\hat{\mathbf{S}})$ from~\eqref{eq:diag}, we obtain
\begin{equation}
    \begin{aligned}
          & \|
    \mathrm{diag} (
      \mathbf{U}_{\mathbf{x}} 
      \mathrm{diag}(\hat{\mathbf{q}}) 
\mathrm{diag}(\boldsymbol{\lambda}_{\mathbf{x}}^{-1/2}) 
      \mathbf{U}_{\mathbf{x}}^\top
    )
    - \mathbf{1}
  \|_2^2  \\
  = & \|  ( {\mathbf{U}}_{\bf x} \circ  {\mathbf{U}}_{\bf x}  )    \text{diag}( {\boldsymbol \lambda}_{\bf x}^{-1/2} )  \hat{\bf q} - \mathbf{1}  \|_2^2,
    \end{aligned}
\end{equation}
where $\circ$ is the Hadamard product.
This equivalence follows from 
$\operatorname{diag}(\mathbf{A} \,\mathrm{diag}(\mathbf{d})\, \mathbf{A}^\top) = (\mathbf{A}\circ \mathbf{A})\,\mathbf{d}$, 
applied with $\mathbf{A}=\mathbf{U}_{\mathbf{x}}$ and $\mathbf{d}=\boldsymbol{\lambda}_{\mathbf{x}}^{-1/2}  \circ  \hat{\mathbf{q}}$.

Defining $\mathbf{W} = 
  (\mathbf{U}_{\mathbf{x}} \circ \mathbf{U}_{\mathbf{x}})
   \mathrm{diag}\bigl(\boldsymbol{\lambda}_{\mathbf{x}}^{-1/2}\bigr)$, the hollowness penalty reduces to 
\begin{equation}
\label{eq:binary_problem_UBQP_undirected}
  \|\mathbf{W} \hat{\mathbf{q}} - \mathbf{1} \|_2^2.
\end{equation}
Minimizing this term over $\hat{\mathbf{q}}  \in  \{-1,1\}^N$
hence results in an unconstrained binary quadratic program (UBQP) \cite{kochenberger2014unconstrained}.

\subsection{Sparsity Constraint}

Substituting the expression for $\hat{\bf S}$ from~\eqref{eq:undirected_structural_opt_2} into $h_{\mathrm{sparse}}(\hat{\mathbf{S}})$ from~\eqref{eq:sparse}, we obtain
\begin{equation}
    \begin{aligned}
        & \|
\mathbf{I}
 - 
\mathbf{U}_{\mathbf{x}} 
\mathrm{diag}(\hat{\mathbf{q}}) 
\mathrm{diag}\bigl(\boldsymbol{\lambda}_{\mathbf{x}}^{-1/2}\bigr) 
\mathbf{U}_{\mathbf{x}}^\top \|_{1} \\
= & \| ( (
\mathrm{diag}\bigl(\boldsymbol{\lambda}_{\mathbf{x}}^{-1/2}\bigr)
\mathbf{U}_{\mathbf{x}}^\top )
\odot
\mathbf{U}_{\mathbf{x}} ) \hat{\bf q} - \mathrm{vec}({\bf I}) \|_1,
    \end{aligned}
\end{equation}
where $\odot$ is the Khatri-Rao product.
Let $\mathbf{M} =  (
\mathrm{diag} (\boldsymbol{\lambda}_{\mathbf{x}}^{-1/2} )
\mathbf{U}_{\mathbf{x}}^\top ) \odot \mathbf{U}_{\mathbf{x}} $ , then the sparsity penalty simplifies to
\begin{equation}\label{eq:sparsity_penalty}
\| \mathbf{M} \hat{\mathbf{q}}- \operatorname{vec}(\mathbf{I}) \|_{1}.
\end{equation}

To minimize this sparsity cost, let $\mathbf{t}\in\mathbb{R}^{N^{2}}_{+}$ be an auxiliary vector that upper bounds every
absolute deviation in~\eqref{eq:sparsity_penalty}.  
Replacing each $\ell_{1}$ term by two linear inequalities then yields the following mixed–integer
linear program (MILP) \cite{floudas2005mixed}

\begin{equation}
    \begin{aligned}
        \min_{\hat{\mathbf{q}},\,\mathbf{t}}
&\quad
\mathbf{1}^{\top}\mathbf{t}\\
\text{s.t.}\;&\quad
-\mathbf{t}\;\le\;
\mathbf{M}\hat{\mathbf{q}}-\operatorname{vec}(\mathbf{I})
\;\le\;\mathbf{t},\\
&\quad
\mathbf{t}\ge\mathbf{0},\\
&\quad
\hat{q}_{i}\in\{-1,1\},
\quad i=1,\dots,N.
\label{prob:MILP_bin}
    \end{aligned}
\end{equation}

\subsection{Conic Mixed--Integer Reformulation}

Combining the hollowness and sparsity costs, our composite objective is
\begin{equation}
\min_{\hat{\mathbf{q}}\in\{-1,1\}^N}
\quad
\|\mathbf{W} \hat{\mathbf{q}} - \mathbf{1} \|_2^2
 + 
\alpha  \|\mathbf{M} \hat{\mathbf{q}} - \operatorname{vec}(\mathbf{I}) \|_1.
\end{equation}

This composite objective combines a squared $\ell_{2}$ penalty with an $\ell_{1}$ penalty, both of which admit conic representations.  Specifically, the squared  Euclidean norm can be embedded in a rotated second–order cone (RSOC) \cite{alizadeh2003second}, while the $\ell_{1}$ norm corresponds to a non-negative orthant cone after the introduction of auxiliary variables.  Hence the entire problem can be written as a {conic mixed–integer program} (CMIP) \cite{atamturk2011lifting}.

To see this formally, introduce a scalar $s\in\mathbb{R}_{+}$ and a vector $\mathbf{t}\in\mathbb{R}^{N^{2}}_{+}$ such that
$\|\mathbf{W}\hat{\mathbf{q}}-\mathbf{1}\|_{2}^2 \leq s
$ and $-\mathbf{t}\leq\mathbf{M}\hat{\mathbf{q}}-\operatorname{vec}(\mathbf{I})\leq\mathbf{t}.$

The quadratic bound is equivalent to the
RSOC condition $\bigl(s,\tfrac12,\mathbf{W}\hat{\mathbf{q}}-\mathbf{1}\bigr) \;\in\;\mathcal{K}_{\mathrm{rSOC}}$ with
\begin{equation}
\label{eq:rsoc_embed}
\mathcal{K}_{\mathrm{rSOC}}
=\;
\Bigl\{(u,v,\mathbf{z})\in\mathbb{R}_{+}^{2}\times\mathbb{R}^{N}:
\,2uv\;\ge\;\|\mathbf{z}\|_{2}^{2}\Bigr\},
\end{equation}
because setting $v=1/2$ yields the scalar inequality
$s\;\ge\;\|\mathbf{W}\hat{\mathbf{q}}-\mathbf{1}\|_{2}^{2}$.  The
element-wise bounds based on $\mathbf{t}$ on the other hand can be formulated as
$\bigl(\mathbf{t},
\mathbf{M}\hat{\mathbf{q}}-\operatorname{vec}(\mathbf{I})\bigr) \; \in \; \mathcal{K}_{1} $ with $\mathcal{K}_{1} $ the $\ell_{1}$ cone given by
\[
\mathcal{K}_{1}
=\Bigl\{(\mathbf{u},\mathbf{v})\in
\mathbb{R}_{+}^{N^{2}}\times\mathbb{R}^{N^{2}}:
-u_{i}\le v_{i}\le u_{i}\Bigr\}.
\]

Collecting these ingredients, the undirected SEM inference task is
equivalent to the following CMIP problem:
\begin{equation}
    \begin{aligned}
        \min_{\hat{\mathbf{q}},\,s,\,\mathbf{t}}
&\quad
s+\alpha\,\mathbf{1}^{\top}\mathbf{t}\\[2pt]
\text{s.t.}\;&\quad
\bigl(s,\tfrac12,\mathbf{W}\hat{\mathbf{q}}-\mathbf{1}\bigr)
\in\mathcal{K}_{\mathrm{rSOC}},\\[2pt]
&\quad
\bigl(\mathbf{t},\mathbf{M}\hat{\mathbf{q}}-\operatorname{vec}(\mathbf{I})\bigr)\in\mathcal{K}_{1},\\[2pt]
&\quad
\hat{q}_{i}\in\{-1,1\},\quad i=1,\dots,N.\label{prob:mir_socp_bin}
    \end{aligned}
\end{equation}

Modern solvers \cite{mosek,gurobi} handle
a CMIP via branch-and-bound enhanced with conic dual bounds and
cutting planes to achieve global optimality.

\section{Covariance Matching for Directed SEM}
\label{sec:directed_extension}

We now turn to the {directed} graph setting, where $\hat{\mathbf{S}}$ (and thus $\hat{\bf H}$) may be non-symmetric. Combining the structural objective $h(\hat{\mathbf{S}})$ from~\eqref{eq:structural} with the covariance matching constraint~\eqref{eq:S_hat_dir} leads to
\begin{equation}
\label{eq:directed_structural_opt_2}
\begin{aligned}
\min_{\hat{\mathbf{V}}: \hat{\mathbf{V}}^\top \hat{\mathbf{V}}=\mathbf{I}}
&\quad
 h_{\mathrm{diag}}(\hat{\mathbf{S}})
+
\alpha h_{\mathrm{sparse}}(\hat{\mathbf{S}})
\\
\text{subject to}
&\quad
\hat{\mathbf{S}}
=
\mathbf{I}
-
\hat{\mathbf{V}}
 \mathrm{diag}(\boldsymbol{\lambda}_{\mathbf{x}}^{-1/2})
 \mathbf{U}_{\mathbf{x}}^\top.
\end{aligned}
\end{equation}

A semidefinite relaxation approach can then be applied using the Russo-Dye theorem~\cite{russo1966note}. While this method often provides a good solution, it requires solving a nontrivial convex optimization problem, which is challenging for large $N$.

\subsection{Riemannian Gradient Descent on the Orthogonal Group}

An alternative is to treat $\hat{\mathbf{V}}$ as a point on the orthogonal group $\mathcal{O}(N)$ and minimize the objective via Riemannian gradient descent. 
We initialise $\hat{\mathbf V}^{(0)}$ with a random orthogonal matrix~\cite{mezzadri2006generate}, i.e., a sample that is uniform on $\mathcal O(N)$ in the sense of the left– and right–invariant Haar measure.  The same initialization is used whenever a random orthogonal matrix is required later (Algorithms~\ref{alg:basinhop_riemann}, \ref{alg:basinhop_riemann_multi}, \ref{alg:sample_rotation}).
To do Riemannian gradient descent, at each iteration we compute the Euclidean (sub)gradient of the objective in $\mathbb{R}^{N\times N}$ and then project it onto the tangent space of the orthonormal manifold~\cite{unitary_opt} (for details on derivative computation in Euclidean space, see Appendix~\ref{appendix:derivative}). Algorithm~\ref{alg:riem_update} outlines a typical procedure.

\begin{algorithm}[H]
\caption{\textsc{RiemannGD}}
\label{alg:riem_update}
\begin{algorithmic}[1]
\State \textbf{Input:} random orthogonal matrix $\hat{\mathbf{V}}^{(0)}\!\in\!\mathcal{O}(N)$~\cite{mezzadri2006generate}
\State \textbf{Output:} approximate minimiser $\hat{\mathbf{V}}^{(+)}$
\State \textbf{Hyperparameters:} step sizes $\{\mu_r\}$, \# iterations $R$
\State $\hat{\mathbf{V}}_{0} \gets \hat{\mathbf{V}}^{(0)}$
\For{$r = 0$ \textbf{to} $R-1$}
  \State $\boldsymbol{\Gamma}_{r}
        \gets \dfrac{\partial\mathcal{J}}{\partial\hat{\mathbf{V}}}
          \bigl(\hat{\mathbf{V}}_{r}\bigr)$ \algorithmiccomment{Euclidean gradient}
  \State $\mathbf{G}_{r}
        \gets \boldsymbol{\Gamma}_{r}\hat{\mathbf{V}}_{r}^{\top}
        -\hat{\mathbf{V}}_{r}\boldsymbol{\Gamma}_{r}^{\top}$ \algorithmiccomment{Project to tangent space}
  \State $\mathbf{G}_{r}\gets \mathbf{G}_{r}/\|\mathbf{G}_{r}\|_{F}$ \algorithmiccomment{Normalize search direction}
  \State $\mathbf{P}_{r}\gets \exp(-\mu_r\mathbf{G}_{r})$ \algorithmiccomment{Exponential map}
  \State $\hat{\mathbf{V}}_{r+1}\gets \mathbf{P}_{r}\hat{\mathbf{V}}_{r}$ \algorithmiccomment{Update}
\EndFor
\State $\hat{\mathbf{V}}^{(+)}\gets\hat{\mathbf{V}}_{R}$
\State \Return $\hat{\mathbf{V}}^{(+)}$
\end{algorithmic}
\end{algorithm}

Here the objective function regarding $\hat{\bf V}$ is
\begin{align}
\label{eq:obj_J}
\mathcal{J}(\hat{\mathbf{V}})
&= h\!\Bigl(\mathbf{I}
   -\hat{\mathbf V}\,\mathrm{diag}(\boldsymbol{\lambda}_{\mathbf x}^{-1/2})
   \mathbf U_{\mathbf x}^{\top}\Bigr) \notag\\
&= \Big\|\mathrm{diag}\!\Bigl(
   \hat{\mathbf V}\,\mathrm{diag}(\boldsymbol{\lambda}_{\mathbf x}^{-1/2})
   \mathbf U_{\mathbf x}^{\top}-\mathbf I\Bigr)\Big\|_F^2 \notag\\
&\quad + \alpha\Big\|
   \hat{\mathbf V}\,\mathrm{diag}(\boldsymbol{\lambda}_{\mathbf x}^{-1/2})
   \mathbf U_{\mathbf x}^{\top}-\mathbf I\Big\|_1 .
\end{align}
In Algorithm~\ref{alg:riem_update}, $\boldsymbol{\Gamma}_{r}$ is the Euclidean gradient of the objective, while $\mathbf{G}_{r}=\boldsymbol{\Gamma}_{r}\hat{\mathbf{V}}_{r}^{\!\top}-\hat{\mathbf{V}}_{r}\boldsymbol{\Gamma}_{r}^{\!\top}$ is its projection onto the Stiefel manifold, hence the true Riemannian gradient.  Directly stepping along $\mathbf{G}_{r}$ would generally destroy orthonormality, so we convert this tangent direction into the rotation $\mathbf{P}_{r}=\exp(-\mu\mathbf{G}_{r})$, which stays on the manifold for any step size $\mu>0$.  Updating $\hat{\mathbf{V}}_{r+1}=\mathbf{P}_{r}\hat{\mathbf{V}}_{r}$ therefore performs descent without leaving the feasible set.  Notably, if the exponential is replaced by its first–order Taylor expansion $\mathbf{I}-\mu\mathbf{G}_{r}$, the rule reduces to ordinary gradient descent. 
Our sole deviation from the standard scheme \cite{unitary_opt} is to normalize $\mathbf{G}_r$ by its Frobenius norm before the exponential map. As $\mathcal{J}$ can vary by orders of magnitude, this makes $\mu$ directly control the geodesic step length and prevents gradient spikes from causing excessive rotations. 
Following common practice in machine learning, we use a time–varying stepsize with exponential decay starting from an initial value $\mu_0$. The concrete choices of decay and initial value used in our runs are reported in Section~\ref{sec:exp}.

However, while convex optimization techniques can circumvent local minima, manifold optimization methods typically do not enjoy such benefits. A straightforward remedy is to try multiple initial guesses, hoping that at least one of them leads to a good local minimum. Since the unitary manifold is pretty large, most random initializations are ineffective though, rendering this naive multi-start approach prohibitively inefficient.

\subsection{Basin-Hopping and Escaping Local Minima}

To escape local minima, we take inspiration from Basin-hopping techniques~\cite{wales1997global}. Specifically, once the algorithm converges to a local minimum, we propose to sample around that local minimum, as summarised in Algorithm~\ref{alg:basinhop_riemann}. The concrete sampling strategy used for these perturbations is detailed in Algorithm~\ref{alg:sample_rotation} at the end of this section. This local sampling procedure raises the possibility of escaping to another region of the manifold that may admit a solution with lower cost, thus offering a middle ground between random restarts and purely local manifold search.

\begin{algorithm}[H]
\caption{Basin–Hopping Scheme}
\label{alg:basinhop_riemann}
\begin{algorithmic}[1]
\State \textbf{Input:} random orthogonal matrix $\hat{\mathbf{V}}^{(0)} \in \mathcal{O}(N)$~\cite{mezzadri2006generate}
\State \textbf{Output:} best solution $\hat{\mathbf{V}}^\star$
\State \textbf{Hyperparameters:} \# samples $K$
\State $\hat{\mathbf{V}}^\star \gets \hat{\mathbf{V}}^{(0)}$ \algorithmiccomment{initialise incumbent}
\For{$k = 1$ \textbf{to} $K$}
  \State $\hat{\mathbf{V}}^{(0)} \gets \textsc{Sample}\!\left(\hat{\mathbf{V}}^\star\right)$ \algorithmiccomment{random local perturbation}
  \State $\hat{\mathbf{V}}^{(+)} \gets \textsc{RiemannGD}\!\left(\hat{\mathbf{V}}^{(0)}\right)$ \algorithmiccomment{local refinement}
  \If{$\mathcal{J}\!\left(\hat{\mathbf{V}}^{(+)}\right) < \mathcal{J}\!\left(\hat{\mathbf{V}}^\star\right)$}
    \State $\hat{\mathbf{V}}^\star \gets \hat{\mathbf{V}}^{(+)}$ \algorithmiccomment{accept if better}
  \EndIf
\EndFor
\State \Return $\hat{\mathbf{V}}^\star$
\end{algorithmic}
\end{algorithm}

Although the above single-threaded basin-hopping loop can usually lead to a satisfactory solution, it may be computationally slow. Noticing that we can generate multiple candidates of $\hat{\mathbf{V}}^{(0)}$ in parallel, we adopt a multi-threaded strategy for greater efficiency as  described in Algorithm~\ref{alg:basinhop_riemann_multi}.

\begin{algorithm}[H]
\caption{Multi-Threaded Basin Hopping}
\label{alg:basinhop_riemann_multi}
\begin{algorithmic}[1]
\State \textbf{Input:} random orthogonal matrix $\hat{\mathbf V}^{(0)}\in\mathcal O(N)$~\cite{mezzadri2006generate}
\State \textbf{Output:} best solution $\hat{\mathbf V}^\star$
\State \textbf{Hyperparameters:} \# cycles $K$, \# samples per cycle $L$
\State $\hat{\mathbf V}^\star \gets \hat{\mathbf V}^{(0)}$ \algorithmiccomment{initialise incumbent}
\For{$k=1$ \textbf{to} $K$}
  \State $\mathcal S \gets \{\textsc{Sample}(\hat{\mathbf V}^\star)\}_{\ell=1}^{L}$ \algorithmiccomment{generate sample set}
  \State $\mathcal R \gets \emptyset$ \algorithmiccomment{initialise refined candidates}
  \ForAll{$\hat{\mathbf V}^{(0)}\in\mathcal S$ \textbf{in parallel}}
    \State $\hat{\mathbf V}^{(+)} \gets \textsc{RiemannGD}(\hat{\mathbf V}^{(0)})$ \algorithmiccomment{local refinement}
    \State $\mathcal R \gets \mathcal R \cup \{\hat{\mathbf V}^{(+)}\}$ \algorithmiccomment{update refined set}
  \EndFor
  \State $\hat{\mathbf V}^{(+)}_{\text{best}} \gets \argmin_{\mathbf V\in\mathcal R}\, \mathcal J(\mathbf V)$ \algorithmiccomment{best of this cycle}
  \If{$\mathcal J(\hat{\mathbf V}^{(+)}_{\text{best}}) < \mathcal J(\hat{\mathbf V}^\star)$}
    \State $\hat{\mathbf V}^\star \gets \hat{\mathbf V}^{(+)}_{\text{best}}$ \algorithmiccomment{accept if better}
  \EndIf
\EndFor
\State \Return $\hat{\mathbf V}^\star$
\end{algorithmic}
\end{algorithm}

We have empirically observed that, due to the sole reliance on the sparsity of the adjacency matrix, the local minimum with the smallest loss need not yield the smallest discrepancy between the recovered and ground–truth adjacencies. 
Consequently, strictly choosing the solution $\hat{\mathbf{V}}^{(+)}_{\text{best}}$ that attains the minimum cost among parallel local searches may discard other candidates with more favorable error.

To address this drawback, we maintain a {candidate set} that is updated after every iteration. As described in Algorithm~\ref{alg:basinhop_riemann_multi_candidates}, rather than drawing perturbations solely from the current best matrix $\hat{\mathbf{V}}^\star$, we sample around a matrix chosen randomly from the candidate set. The newly sampled matrices are merged with the existing candidate set, and we apply \textsc{RiemannGD} to each matrix in the combined set. Retaining earlier candidates prevents high-quality solutions from being discarded and encourages diversity, even when some candidates do not currently attain the lowest cost.

In updating the candidate set, we do not simply take the $L'$ matrices with the smallest objective values, as this would result in many nearly identical solutions. Instead, a diversity criterion is enforced so that only sufficiently distinct candidates are retained. The detailed procedure is given in Appendix~\ref{appendix:candidate_update}.

\begin{algorithm}[H]
\caption{Multi-Threaded Basin Hopping with Candidate Set}
\label{alg:basinhop_riemann_multi_candidates}
\begin{algorithmic}[1]
\State \textbf{Input:} random candidate set $\mathcal C^{(0)}\subseteq\mathcal O(N)$~\cite{mezzadri2006generate}
\State \textbf{Output:} best solution $\hat{\mathbf V}^\star$
\State \textbf{Hyperparameters:} \# cycles $K$, \# samples per cycle $L$, candidate-set size $L'$
\State $\mathcal C \gets \mathcal C^{(0)}$ \algorithmiccomment{initialise candidate set}
\For{$k=1$ \textbf{to} $K$}
  \State $\mathcal S \gets \emptyset$ \algorithmiccomment{initialise sample set}
  \For{$\ell=1$ \textbf{to} $L$}
     \State $\hat{\mathbf{V}}_{c} \gets \textsc{RandomChoice}(\mathcal{C})$  \algorithmiccomment{select a candidate}
     \State $\hat{\mathbf V}^{(0)} \gets \textsc{Sample}(\hat{\mathbf{V}}_{c} )$ \algorithmiccomment{apply  perturbation}
     \State $\mathcal S \gets \mathcal S \cup \{\hat{\mathbf V}^{(0)} \}$ \algorithmiccomment{update sample set}
  \EndFor
  \State $\mathcal R \gets \emptyset$
  \ForAll{$\hat{\mathbf V}^{(0)}\in \mathcal C\cup\mathcal S$ \textbf{in parallel}}
     \State $\hat{\mathbf V}^{(+)} \gets \textsc{RiemannGD}(\hat{\mathbf V}^{(0)})$
     \algorithmiccomment{local refinement}
     \State $\mathcal R \gets \mathcal R \cup \{\hat{\mathbf V}^{(+)}\}$\algorithmiccomment{update refined set}
  \EndFor
  \State $\mathcal{C} \gets \left\{ \hat{\mathbf{V}} \in \mathcal{R} \,:\, L' \ \text{most diverse and low cost} \right\}$%
  
\hfill$\triangleright$~update candidate set
\EndFor
\State $\hat{\mathbf V}^\star \gets \argmin_{\hat{\mathbf V}\in\mathcal C}\,\mathcal J(\hat{\mathbf V})$  \algorithmiccomment{best in candidate set}
\State \Return $\hat{\mathbf V}^\star$
\end{algorithmic}
\end{algorithm}

\subsection{Sampling on the Orthogonal Group}

A key step in our algorithms involves sampling orthogonal matrices in a manner that is not uniformly distributed over the orthogonal group, but instead biased toward moderate rotations. Equivalently, given $\mathbf{V}$ we perturb it by a moderate random rotation: sample an orthogonal $\mathbf{U}'$ concentrated near $\mathbf{I}$ and set $\mathbf{V}'=\mathbf{U}'\mathbf{V}$. Thus the problem reduces to specifying how to generate $\mathbf{U}'$ near $\mathbf{I}$.

Starting from a random orthogonal matrix $\mathbf{U}\in\mathcal{O}(N)$ ~\cite{mezzadri2006generate}, write
\begin{equation}
    \mathbf{U}=\exp\!\bigl(\log(\mathbf{U})\bigr).
\end{equation}
When $\det(\mathbf{U})>0$, the set $\{\exp(\tau\,\log(\mathbf{U})):\,\tau\in[0,1]\}$ is the shortest geodesic connecting $\mathbf{I}$ and $\mathbf{U}$~\cite{moakher2002means}. Hence, drawing $\tau\in[0,1]$ and setting $\mathbf{U}'=\exp(\tau\,\log(\mathbf{U}))$ samples exactly along this geodesic, shrinking all rotation angles by a factor $\tau$. When $\det(\mathbf{U})<0$, we flip the sign of the first column of $\mathbf{U}$ to make the logarithm real, compute $\mathbf{U}'=\exp(\tau\,\log(\mathbf{U}))$, and then again flip the first column of $\mathbf{U}'$ so that $\det(\mathbf{U}')<0$ remains possible. This preserves orthogonality and allows sampling from either component of $\mathcal{O}(N)$. The complete procedure is summarized in Algorithm~\ref{alg:sample_rotation}.

\begin{algorithm}[H]
\caption{\textsc{Sample}}
\label{alg:sample_rotation}
\begin{algorithmic}[1]
\State \textbf{Input:} orthogonal matrix $\mathbf{V}$
\State \textbf{Output:} sampled orthogonal matrix $\mathbf{V}'$
\State \textbf{Hyperparameters:} $\tau_{\min},\tau_{\max}$ with $0\le\tau_{\min}\le\tau_{\max}\le1$
\State $\mathbf{U}_{\text{raw}} \gets \textsc{RandomOrthogonal}({\bf I})$ \algorithmiccomment{select random orthogonal matrix~\cite{mezzadri2006generate}}
\State $\mathbf{U} \gets \mathbf{U}_{\text{raw}}$
\If{$\det(\mathbf{U}) < 0$}
  \State $\mathbf{U}_{:,1} \gets -\,\mathbf{U}_{:,1}$ \algorithmiccomment{force $\det(\mathbf{U})=+1$ }
\EndIf
\State $\tau \sim \mathcal{U}(\tau_{\min},\tau_{\max})$ \algorithmiccomment{geodesic step size}
\State $\mathbf{U}' \gets \exp(\tau\  \log\mathbf{U})$ \algorithmiccomment{step $\tau$ along geodesic}
\If{$\det(\mathbf{U}_{\text{raw}}) < 0$}
  \State $\mathbf{U}'_{:,1} \gets -\,\mathbf{U}'_{:,1}$ \algorithmiccomment{restore original $O(N)$ parity}
\EndIf
\State $\mathbf{V}' \gets \mathbf{U}'\,\mathbf{V}$ \algorithmiccomment{apply perturbation}
\State \Return $\mathbf{V}'$
\end{algorithmic}
\end{algorithm}

\section{Experiments}
\label{sec:exp}

We evaluate four settings: (i) undirected synthetic graphs, (ii) DAGs with established baselines, (iii) general directed (cyclic) graphs, and (iv) real data.

Across all simulations, we report the average normalised squared error (NSE)
\begin{equation}
\operatorname{NSE}\bigl(\mathbf{S},\mathbf{S}^{\!*}\bigr)
=
\frac{\|\mathbf{S}^{\!*}-\mathbf{S}\|_F^{2}}{\|\mathbf{S}\|_F^{2}},
\end{equation}
where $\mathbf{S}$ is the ground-truth adjacency and ${\mathbf{S}}^*$ its estimate. 

For the $N{=}11$ T–cell protein–expression data of \cite{sachs2005causal} (14 experimental conditions; $707$–$927$ samples each), we follow \cite{amendola2020structure} to handle non-Gaussianity. We then fit our estimator separately to each condition and compare the resulting network with that obtained by \cite{amendola2020structure}.

\subsection{Simulation on Undirected Graphs}
We first focus on the symmetric setting.  For each graph size $N \in \{20,40,60,80,100\}$, we draw an undirected topology with $M=2N$ edges and assign weights uniformly i.i.d. from $[-1,-0.1]\cup[0.1,1]$. 
Our estimator is benchmarked against the {SpecTemp}~\cite{spectemp} and {SigMatch}~\cite{aragam2015learning} baselines with $T=1000$ and $T \to \infty$.  Because {SpecTemp} is only applicable when all edge weights are nonnegative, for this method, the edge weights are sampled from $[0.1,1]$.
For every $N$, we generate ten graphs, apply each competing method, and report the {average NSE}

For the SpecTemp baseline, we solve
   
\begin{equation}
\begin{aligned}
\min_{\hat{\mathbf S},\,\mathbf q}\;
  & \alpha \|\hat{\mathbf S}\|_{1} +  \|\hat{\mathbf S}-\hat{\mathbf S}'\|_F \\
\text{subject to}\quad
  & \hat{\mathbf S}'=\hat{\mathbf U}_{{\bf x}} \operatorname{diag}(\mathbf q)\hat{\mathbf U}^\top_{{\bf x}},\;
    \hat{\mathbf S}\in\mathscr S,
\end{aligned}
\end{equation}
where $\hat{\mathbf U}$ is the eigenvector matrix of the sample covariance matrix ${\bf C}_{\bf x}$.  
The feasible set is given by $\mathscr S=\{\hat{\mathbf S}\in\mathbb R^{N\times N}\mid
\hat{\mathbf S}^\top=\hat{\mathbf S},\;
\hat{\mathbf S}\ge 0,\;
\operatorname{diag}(\hat{\mathbf S})=0,\;
\hat{\mathbf S}\mathbf 1\ge \mathbf 1\}$ follows~\cite{buciulea2025polynomial}.  
Here, the distance bound of $\|\hat{\mathbf S}-\hat{\mathbf S}'\|_F$ is included in the objective to avoid solver failure.  
Because SpecTemp has a scale ambiguity, we rescale both ${\mathbf S}^\star$ and $\hat{\mathbf S}$ by their Frobenius norms before computing the NSE.

SigMatch is obtained by solving
\begin{equation}
\min_{\hat{\mathbf S}}\;
    \tfrac1T\|\mathbf X-\hat{\mathbf S}\mathbf X\|_F^2
    +\alpha\|\hat{\mathbf S}\|_{1}
   =
  \operatorname{tr}(\hat{\mathbf S}\mathbf C_{\bf x}\hat{\mathbf S}^\top)
  -2\,\operatorname{tr}(\hat{\mathbf S}\mathbf C_{\bf x})
  +\alpha\|\hat{\mathbf S}\|_{1},
\end{equation}
where ${\bf C_x}=\tfrac1T\mathbf X\mathbf X^{\top}$ for $T=1000$, and which equals $\bf \Sigma_x$  when $T\to\infty$.

For CovMatch, we solve Problem~\eqref{prob:mir_socp_bin}.
For all methods, $\alpha$ controls the importance of sparsity.
To select an appropriate $\alpha$, for each $T$ and each method, we perform a grid search over $\alpha$ and select the best value before testing on the generated graphs.

\begin{figure}[H]
  \centering
  \hspace*{4mm}
  \includegraphics[width=0.9\linewidth]{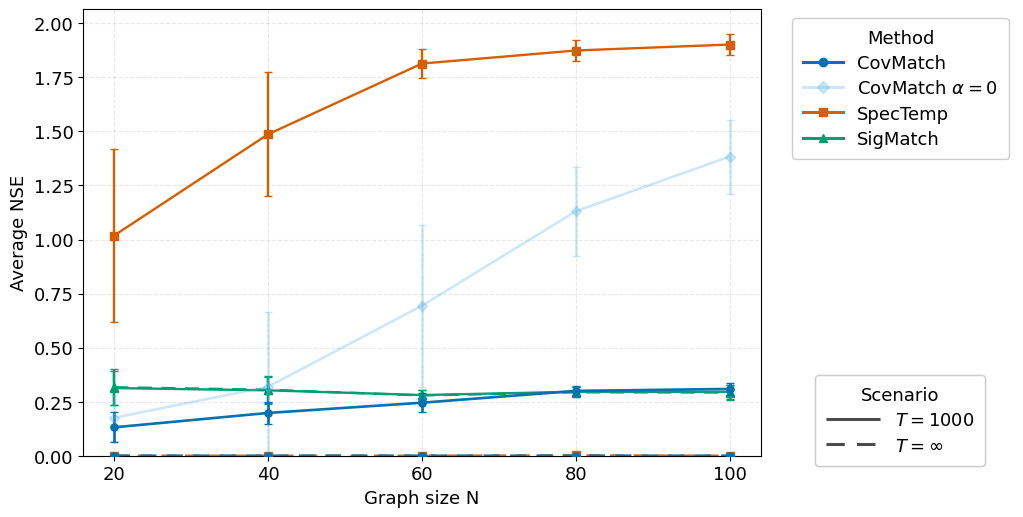}
  \caption{Average NSE versus number of nodes}
  \label{fig:undir_exp}
\end{figure}

Figure~\ref{fig:undir_exp} reports the average NSE versus graph size $N$ for two sample sizes, $T{=}1000$ (solid) and $T{\to}\infty$ (dashed). We compare CovMatch, SpecTemp, SigMatch, and {CovMatch with $\alpha{=}0$}, which removes the sparsity term and coincides with the ICASSP variant of CovMatch that enforces only hollowness~\cite{ICASSP_covmatch}.
Because its accuracy hinges on reliably recovering the eigenbasis of the sample covariance, stable estimation demands a very large observation horizon; consequently, for moderate $T$ its error surpasses that of {SigMatch}.  By introducing an explicit sparsity prior, our method avoids producing overly dense graphs, achieving markedly lower errors across practical horizons while preserving the asymptotic guarantee that the error tends to zero as $T \to \infty$.

\subsection{Simulation on DAGs}
\label{subsec:exp_dag}

For general directed graphs, there is no widely acknowledged baseline, because most existing algorithms specialize in particular structural assumptions.  To provide a meaningful comparison, we therefore first restrict attention to DAGs and benchmark against the two very popular causal–discovery methods, {NOTEARS}~\cite{DAG_notears} and {DAGMA}~\cite{bello2022dagma}.   We subsequently evaluate our estimator on cyclic directed graphs as well, see Subsection~\ref{subsec:sim_random_directed}.

Following the public implementation of {NOTEARS}, we generate ten DAGs for every graph size $N \in \{20,40,60,80,100\}$.  
Each DAG is obtained by first drawing a random undirected Erdős–Rényi graph with edge probability \(p = 2/N\). Then we only maintain the lower‑triangular part of the adjacency matrix and randomly permute the vertex labels to eliminate any inherent topological ordering. 
Edge weights are then uniformly drawn from the default range $[-2,-0.5]\cup[0.5,2]$ used in NOTEARS.

Across all simulation experiments that apply CovMatch to directed graphs, we set the sparsity weight to $\alpha = 10^{-2}$. The role of $\alpha$ is simply to promote sparsity once hollowness (zero diagonal) is enforced. In Algorithm~\ref{alg:basinhop_riemann_multi_candidates}, we set $L=L'=64$, matching the number of CPU cores, and $K=200$.  To improve efficiency, if the objective does not change for more than $35$ iterations, we terminate early and return the current best candidate. In Algorithm~\ref{alg:sample_rotation}, we use $\tau_{\min}=0.5$ and $\tau_{\max}=0.8$, while Algorithm~\ref{alg:riem_update} contains $R = 10^{4}$ iterations.   For the \textsc{RiemannGD} launched from the random sample set $\mathcal{S}$ the Riemannian step size $\mu_{r}$ decays exponentially from $2\times10^{-2}$ to $4\times10^{-3}$; when the start point is taken from the previous candidate set $\mathcal{C}$ it decays from $10^{-3}$ to $2\times10^{-4}$.  This choice yields errors that are close to zero. In fact, choosing even smaller final step sizes can lower the error further at the expense of additional runtime.

For NOTEARS and DAGMA, the $\ell_{1}$ regularization weight $\alpha$ is selected separately for each $N$ and $T$ by grid search, and the chosen value is then applied when testing on the newly generated graphs.

\begin{figure}[H]
    \centering
    \includegraphics[width=8.3cm]{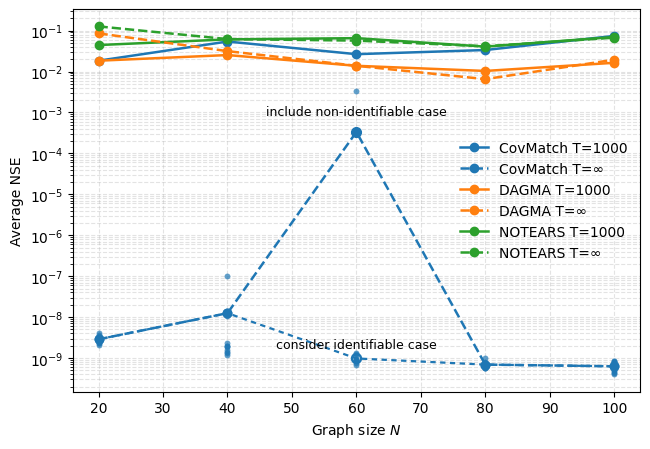}
    \caption{DAG benchmark. Average NSE (log scale) for CovMatch, NOTEARS, and DAGMA. Solid curves: $T=1000$; dotted markers: $T\to\infty$. The plot overlays results excluding and including a single non‑identifiable instance at $N=60$ (annotated in the figure).}
    \label{fig:dag-compare-both}
\end{figure}

\paragraph{Results}
Figure~\ref{fig:dag-compare-both} summarizes the results. 
In the asymptotic limit ($T\to\infty$), CovMatch drives the NSE down to (near) machine precision, whereas NOTEARS and DAGMA plateau around $10^{-2}\!-\!10^{-1}$. This consistency is attained without imposing acyclicity, highlighting the flexibility of our framework.

At the finite horizon $T=1000$, DAGMA achieves the smallest NSE across most $N$. CovMatch is competitive with NOTEARS: it performs better for smaller graphs ($N\le 60$) and is comparable or slightly worse for larger $N$. The finite‑sample gap is attributable to our reliance on the singular vectors of the sample covariance—estimating these accurately becomes harder as the dimensionality increases, so larger networks require longer time series to reach the same accuracy.

The comparison is conservative in that we report the raw output of CovMatch, whereas both NOTEARS and DAGMA perform a post‑hoc pruning step that removes edges with $|w_{ij}|<0.3$. Applying the same pruning to CovMatch further reduces its NSE (not shown).

\paragraph{Non‑identifiable case}
While generating random DAGs, we observed one anomalous instance at $N=60$ in which CovMatch attained a strictly lower objective value than the true adjacency. This indicates that a sparser graph can reproduce the same covariance, revealing fundamental non‑identifiability. Including this instance yields a visible spike at $N=60$ in the averages, but outside this case, the performance is stable across $N$. Even in the presence of such ambiguity, CovMatch attains the smallest error as $T\to\infty$ and often recovers graphs close to the ground truth.

\subsection{Simulation on Cyclic Directed Graphs}
\label{subsec:sim_random_directed}

This experiment investigates how CovMatch scales with network size on general (cyclic) directed graphs. 
For each target size \(N\in\{20,\,40,\,60,\,80,\,100\}\) we generate a graph by selecting \(M=2N\) directed edges uniformly at random, without replacement. 
Every selected edge receives a weight drawn uniformly from the union of intervals \([-1,-0.1]\cup[0.1,1]\).
If the resulting graph happens to be acyclic, it is discarded and regenerated so that the test set contains only challenging non-DAG topologies. We report average NSE across 10 graphs at $T\!\to\!\infty$ and at $T=1000$.

\begin{figure}[H]
    \centering
    \includegraphics[width=8.3cm]{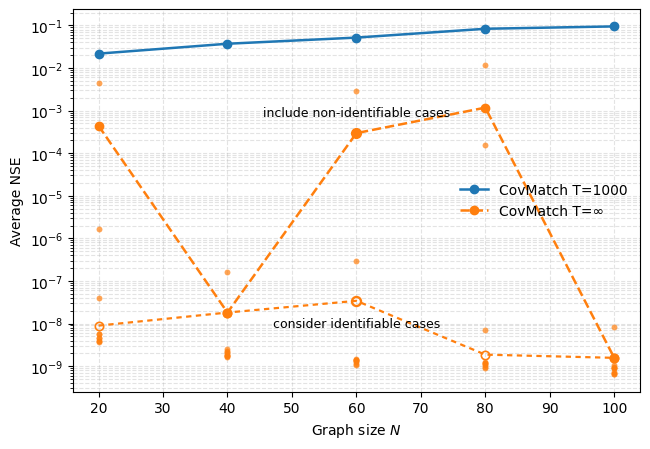}
    \caption{Cyclic directed graphs. Average NSE (log scale) for CovMatch.
    Solid curve: $T=1000$; dashed markers: $T\to\infty$.
    The plot overlays results \emph{excluding} and \emph{including} a small number of non‑identifiable instances (annotated in the figure).}
    \label{fig:directed-compare-both}
\end{figure}

\paragraph{Results}
Figure~\ref{fig:directed-compare-both} shows that, as $T\to\infty$, CovMatch achieves near‑machine‑precision error across all sizes once anomalous non‑identifiable cases are excluded.
When $T=1000$, the average NSE increases with dimension (from roughly $2\times10^{-2}$ at $N=20$ to around $10^{-1}$ at $N=100$), which is expected: for fixed $T$, the ratio $T/N$ shrinks (from $50$ to $10$), making the sample covariance noisier and the problem harder.

\paragraph{Non‑identifiable cases}
As with DAGs, we occasionally observe graphs for which the CovMatch objective at the estimate is \emph{lower} than at the ground truth, indicating that a sparser adjacency reproduces the same covariance; this reflects fundamental non‑identifiability rather than estimation failure.
Including such instances produces local spikes in the $T\to\infty$ curve, whereas the identifiable cases remain consistently close to zero error.

\subsection{Real data}

We evaluate the proposed estimator on the $N=11$ T–cell protein–expression data originally assembled by \cite{sachs2005causal}.  The collection contains $14$ data sets, each recorded under a distinct experimental condition, with sample sizes ranging from $707$ to $927$.  Departures from linear–Gaussianity are handled exactly as in \cite{amendola2020structure}: every variable is treated as a monotone transformation of an underlying Gaussian factor, and the sample covariance in the Gaussian log-likelihood is replaced by the bias-corrected Kendall correlation matrix whose entry $(i,j)$ equals $\sin\!\bigl(\tfrac{\pi}{2}\,\hat{\tau}_{ij}\bigr)$, where $\hat{\tau}_{ij}$ is Kendall’s rank correlation.

For each of the fourteen data sets, we run our estimator independently. Because the copula–transformed covariance does not strictly adhere to the (linear) SEM, we place greater emphasis on sparsity and set \(\alpha = 0.1\).  
We retain only the twenty edges with the largest absolute weights and discard the rest. An edge is then included in the consensus graph if it survives this pruning in at least five of the fourteen runs.
The consensus network of \cite{sachs2005causal} serves as the ground truth; Fig.~\ref{fig:tcell-compare} shows this ground truth (left) and our estimate (right).

\begin{figure}[htbp]
  \centering
  \subfloat[Consensus network from \cite{sachs2005causal} \label{fig:tcell-gt}]{
    \includegraphics[width=0.45\linewidth]{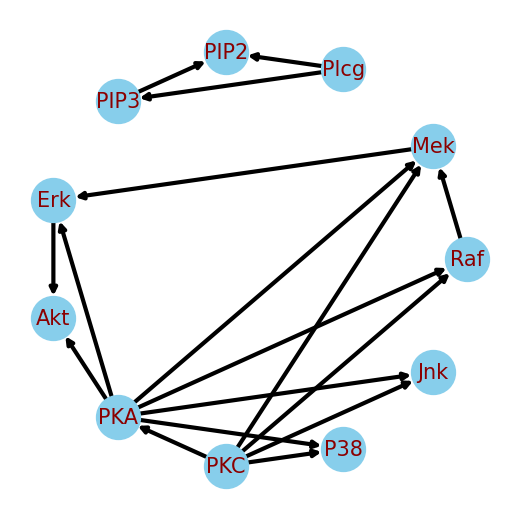}
  }\hfill
  \subfloat[Graph estimated by our method (edge frequency $\ge 6$)\label{fig:tcell-est}]{
    \includegraphics[width=0.45\linewidth]{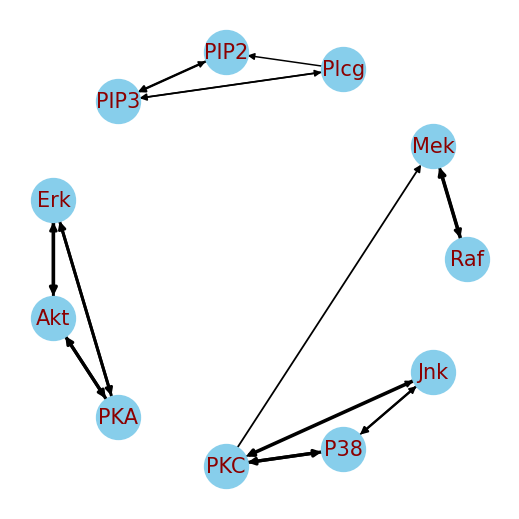}
  }
  \caption{Left: reference consensus network. Right: estimate obtained with the proposed approach.}
  \label{fig:tcell-compare}
\end{figure}
\begin{figure}[htbp]
  \centering
  \includegraphics[width=0.58\linewidth]{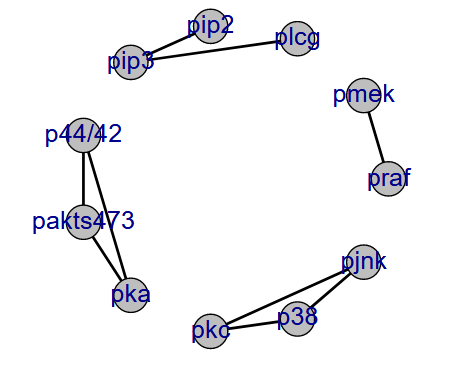}
  \caption{Graph reconstructed by the method of \cite{amendola2020structure}.}
  \label{fig:tcell-amendola}
\end{figure}

Although the protein–expression variables are clearly non-Gaussian, the copula transformation retains sufficient second-order structure for reliable recovery.  The graph produced by our method (Fig.~\ref{fig:tcell-est}) is biologically more plausible than the network obtained with the procedure of \cite{amendola2020structure} shown in Fig.~\ref{fig:tcell-amendola}, since our estimate contains more correctly recovered adjacencies.

\section{Conclusion}

In this study, we presented {CovMatch}, a unifying framework that reconstructs network structure by re-expressing the sample covariance under analytically tractable constraints and encoding prior graph knowledge directly in the optimization objective.  Using the linear structural equation model as a running example, we derived the optimization problem for both undirected and directed graphs.  In the undirected setting, the problem can be solved with off‑the‑shelf solvers, whereas in the directed setting, we propose an algorithm based on Riemannian gradient descent.

Comprehensive experiments on synthetic data reveal that {CovMatch} attains near-zero NSE when the underlying adjacency is sparse.  Remarkably, in the directed case, our method achieves lower asymptotic error than widely used DAG-specific algorithms such as {NOTEARS} and {DAGMA}, despite not imposing an acyclicity prior.  Real-world validation on the T–cell protein‐expression benchmark further confirms the practical utility of the approach.

\bibliographystyle{IEEEbib}
\bibliography{refs}

\appendix

\subsection{Signal Matching Limitation}
\label{appendix:naive_proof}

Consider the signal matching approach
\begin{equation}
\label{eq:naive_appendix}
\min_{\hat{\mathbf{S}}}
\frac{1}{T}\|\mathbf{X} - \hat{\mathbf{S}}\,\mathbf{X}\|_F^{2}
\quad\text{subject to}\quad
\hat{\mathbf{S}}=\hat{\mathbf{S}}^\top,\,
\operatorname{diag}(\hat{\mathbf{S}})=\mathbf{0}.
\end{equation}
Assume the data follow the SEM, i.e., for each column $\mathbf x_t$ of $\mathbf X=[\mathbf x_1,\ldots,\mathbf x_T]$,
$\mathbf x_t=\mathbf S\,\mathbf x_t+\mathbf e_t$ with $\mathbf S$ hollow and symmetric, and $\mathbf e_t\sim\mathcal N(\mathbf 0,\mathbf I)$ i.i.d.\ across $t$.
One might hope that a large sample size would make \(\hat{\mathbf{S}}=\mathbf{S}\) the optimal solution. However, we show below that the true \(\mathbf{S}\) cannot minimize \eqref{eq:naive_appendix} even $T$ goes to infinity.

\begin{proof}

To justify our claim, we first examine the objective function as $T\to\infty$.
Defining the sample covariance as $\mathbf C_{\mathbf x}:=\tfrac1T\,\mathbf X\mathbf X^\top$, we obtain the following trace form for the objective:
\begin{equation}
    \frac{1}{T}\bigl\|\mathbf X-\hat{\mathbf S}\mathbf X\bigr\|_F^2
= \operatorname{tr}(\mathbf C_{\mathbf x})
-2\,\operatorname{tr}(\hat{\mathbf S}\mathbf C_{\mathbf x})
+\operatorname{tr}(\hat{\mathbf S}\mathbf C_{\mathbf x}\hat{\mathbf S}^\top).
\end{equation}

When $T$ goes to infinity, by the SEM relationship, we have  $\mathbf C_{\mathbf x}\to(\mathbf I-\mathbf S)^{-2}$. We therefore define
\begin{equation}
\label{eq:def_f_signal_match}
\begin{aligned}
f(\hat{\mathbf S})
&:= \lim_{T\to\infty}\frac{1}{T}\bigl\|\mathbf X-\hat{\mathbf S}\mathbf X\bigr\|_F^2 - \operatorname{tr}(\mathbf C_{\mathbf x}) \\
& =-2\,\operatorname{tr}(\hat{\mathbf S}(\mathbf I-\mathbf S)^{-2})
+\operatorname{tr}(\hat{\mathbf S}(\mathbf I-\mathbf S)^{-2}\hat{\mathbf S}^\top)
\end{aligned}
\end{equation}
Suppose, for the sake of argument, that \(\hat{\mathbf{S}}=\mathbf{S}\) is indeed optimal for $f(\hat{\mathbf S})$. We then aim to derive a contradiction by examining the derivative structure.

From standard matrix calculus (see \cite{matrix_cookbook}), the derivative under a symmetry constraint becomes
\begin{equation}
\label{eq:df_structured}
\frac{d f}{d \hat{\mathbf{S}}}
\;=\;
\frac{\partial f}{\partial \hat{\mathbf{S}}}
\;+\;
\frac{\partial f}{\partial \hat{\mathbf{S}}^\top}
\;-\;
\operatorname{Diag}\!\bigl(
\tfrac{\partial f}{\partial \hat{\mathbf{S}}}
\bigr).
\end{equation}
Clearly, if any off-diagonal entry of $\frac{\partial f}{\partial \hat{\mathbf{S}}}
+
\frac{\partial f}{\partial \hat{\mathbf{S}}^\top}$ is nonzero, we could move a small step along the corresponding symmetric zero-diagonal direction to further decrease $f$.
Hence, $\frac{\partial f}{\partial \hat{\mathbf{S}}}
+
\frac{\partial f}{\partial \hat{\mathbf{S}}^\top}$ must be a diagonal matrix.

From $f(\hat{\mathbf S})
= -2\,\operatorname{tr}\!\bigl(\hat{\mathbf S}(\mathbf I-\mathbf S)^{-2}\bigr)
+ \operatorname{tr}\!\bigl(\hat{\mathbf S}(\mathbf I-\mathbf S)^{-2}\hat{\mathbf S}^\top\bigr)$
and the symmetry of $(\mathbf I-\mathbf S)^{-2}$, the partial derivative is
\begin{equation}
\label{eq:df_partial_general}
\frac{\partial f}{\partial \hat{\mathbf S}}
= -2(\mathbf I-\mathbf S)^{-2}
+ 2\,\hat{\mathbf S}(\mathbf I-\mathbf S)^{-2}
= 2(\hat{\mathbf S}-\mathbf I)\,(\mathbf I-\mathbf S)^{-2}.
\end{equation}
Evaluating at $\hat{\mathbf S}=\mathbf S$ gives
\begin{equation}
\label{eq:df_partial_atS}
\left.\frac{\partial f}{\partial \hat{\mathbf{S}}}
+
\frac{\partial f}{\partial \hat{\mathbf{S}}^\top}\right|_{\hat{\mathbf S}=\mathbf S}
= -4(\mathbf I-\mathbf S)(\mathbf I-\mathbf S)^{-2}
= -4(\mathbf I-\mathbf S)^{-1}.
\end{equation}

For this to be a diagonal matrix, \((\mathbf{I}-\mathbf{S})^{-1}\) must be a diagonal matrix, which means the true \(\mathbf{S}\) should also be diagonal. However, since we assume no self loops, this means only the trivial (zero) graph can be a solution.

This contradicts the goal of recovering a meaningful adjacency matrix, proving that naive node-domain fitting fails to identify \(\mathbf{S}\) even with infinite samples.
\end{proof}

\subsection{Maximum Likelihood Estimation for Covariance}
\label{appendix:MLE}

We can interpret the constraint in the CovMatch problem~\eqref{eq:role_reversal} as the solution to the unconstrained classical least squares problem on the left-hand side of~\eqref{eq:role_reversal}. We show now that this constraint can also be viewed as the solution to the unconstrained maximum likelihood problem, where we minimize 
\begin{equation}
\label{eq:f_Sigma_hat}
f(\hat{\boldsymbol{\Sigma}}_{\mathbf{x}})
=
\log \det \bigl(\hat{\boldsymbol{\Sigma}}_{\mathbf{x}}\bigr)
\;+\;
\mathrm{tr} \bigl(
\hat{\boldsymbol{\Sigma}}_{\mathbf{x}}^{-1}
\,\mathbf{C}_{\mathbf{x}}
\bigr).
\end{equation}

Using standard matrix calculus~\cite{matrix_cookbook} and the symmetry of \(\hat{\boldsymbol{\Sigma}}_{\mathbf{x}}\), we obtain

\begin{equation}
\label{eq:derivatives_logdet_tr}
\begin{aligned}
& \frac{\partial}{\partial \hat{\boldsymbol{\Sigma}}_{\mathbf{x}}}
\Bigl(\log \det \bigl( \hat{\boldsymbol{\Sigma}}_{\mathbf{x}}\bigr) \Bigr)
\;=\;
\hat{\boldsymbol{\Sigma}}_{\mathbf{x}}^{-1},
\\
& \frac{\partial}{\partial \hat{\boldsymbol{\Sigma}}_{\mathbf{x}}}
\Bigl(\mathrm{tr}\bigl(\hat{\boldsymbol{\Sigma}}_{\mathbf{x}}^{-1}
\,\mathbf{C}_{\mathbf{x}}\bigr)\Bigr)
\;=\;
-\,\hat{\boldsymbol{\Sigma}}_{\mathbf{x}}^{-1}
\,\mathbf{C}_{\mathbf{x}}
\,\hat{\boldsymbol{\Sigma}}_{\mathbf{x}}^{-1}.
\end{aligned}
\end{equation}

Thus,
\begin{equation}
\label{eq:df_Sigma_hat}
\frac{\partial f}{\partial \hat{\boldsymbol{\Sigma}}_{\mathbf{x}}}
\;=\;
\hat{\boldsymbol{\Sigma}}_{\mathbf{x}}^{-1}
\;-\;
\hat{\boldsymbol{\Sigma}}_{\mathbf{x}}^{-1}
\,\mathbf{C}_{\mathbf{x}}
\,\hat{\boldsymbol{\Sigma}}_{\mathbf{x}}^{-1}.
\end{equation}
Setting \(\partial f / \partial \hat{\boldsymbol{\Sigma}}_{\mathbf{x}} = \mathbf{0}\)  implies
\begin{equation}
\label{eq:swap_minimizer}
\hat{\boldsymbol{\Sigma}}_{\mathbf{x}}
=
\mathbf{C}_{\mathbf{x}}
\end{equation}
So both the least squares cost and the maximum likelihood cost on the left-hand side of~\eqref{eq:role_reversal} lead to the same CovMatch problem on the right-hand side of~\eqref{eq:role_reversal}.

\subsection{Covariance Matching for Colored Noise}
\label{appendix:covmatch_general}
We now move to the general scenario in which \(\boldsymbol{\Sigma}_{\mathbf{e}}\) does not need to be the identity matrix. 

For an undirected graph, although one could directly match \(\hat{\mathbf{H}}\,\boldsymbol{\Sigma}_{\mathbf{e}}\,\hat{\mathbf{H}}^{\top}\) to \(\mathbf{C}_{\mathbf{x}}\), solving that problem can be cumbersome. Thus, rather than matching $\hat{\mathbf{H}}\,\boldsymbol{\Sigma}_{\mathbf{e}}\,\hat{\mathbf{H}}^{\top}$ to $\mathbf{C}_{\mathbf{x}}$ directly, we can enforce
\begin{equation}
\label{eq:Hhat_match}
    (\hat{\mathbf H}\,\boldsymbol\Sigma_{\mathbf e})^2 = \mathbf C_{\mathbf x}\,\boldsymbol\Sigma_{\mathbf e},
\end{equation}
which mirrors the reasoning used in the simpler identity-noise case.

For the left-hand side of \eqref{eq:Hhat_match}, we introduce \(\hat{\mathbf{U}}\) and \(\hat{\boldsymbol{\lambda}}\) via the eigenvalue decomposition (EVD):
\begin{equation}
\label{eq:Hhat_Sigmae}
\hat{\mathbf{H}}\,\boldsymbol{\Sigma}_{\mathbf{e}}
=
\hat{\mathbf{U}}\,
\mathrm{diag}(\hat{\boldsymbol{\lambda}})
\,
\hat{\mathbf{U}}^{-1}.
\end{equation}
We then obtain
\begin{equation}
(\hat{\mathbf{H}}\,\boldsymbol{\Sigma}_{\mathbf{e}})^2
=
\hat{\mathbf{U}}\,
\mathrm{diag}(\hat{\boldsymbol{\lambda}}^2)
\,
\hat{\mathbf{U}}^{-1}.
\end{equation}
For the right-hand side of \eqref{eq:Hhat_match}, the EVD leads to
\begin{equation}
\mathbf{C}_{\mathbf{x}}\,\boldsymbol{\Sigma}_{\mathbf{e}}
=
\mathbf{U}_{\mathbf{xe}}\,
\mathrm{diag}(\boldsymbol{\lambda}_{\mathbf{xe}})
\,
\mathbf{U}_{\mathbf{xe}}^{-1}.
\end{equation}
Note that $\mathbf{U}_{\mathbf{xe}}$ is not necessarily orthogonal because $\mathbf{C}_{\mathbf{x}}\boldsymbol\Sigma_{\mathbf e}$ is generally non-symmetric. 
By setting \(\hat{\mathbf{U}}=\mathbf{U}_{\mathbf{xe}}\) and enforcing \(\hat{\boldsymbol{\lambda}}^2=\boldsymbol{\lambda}_{\mathbf{xe}}\), we once again face a sign ambiguity. Introducing \(\hat{\mathbf{q}}\in\{-1,1\}^N\) and writing
\begin{equation}
  \hat{\boldsymbol{\lambda}}
  = \operatorname{diag}(\hat{\mathbf{q}})\,
    \boldsymbol{\lambda}_{\mathbf{xe}}^{1/2}
\end{equation}
resolves this ambiguity. Substituting $ \hat{\boldsymbol{\lambda}}$ and $\hat{\mathbf{U}}$ back into \eqref{eq:Hhat_Sigmae} gives
\begin{equation}
\label{eq:HC}
  \hat{\mathbf{H}}
  = \mathbf{U}_{\!\mathbf{xe}}\,
    \operatorname{diag}(\hat{\mathbf{q}})\,
    \operatorname{diag}\bigl(\boldsymbol{\lambda}_{\mathbf{xe}}^{1/2}\bigr)\,
    \mathbf{U}_{\!\mathbf{xe}}^{-1}\,
    \boldsymbol{\Sigma}_{\mathbf{e}}^{-1},
\end{equation}
resulting in
\begin{equation}
  \hat{\mathbf{S}}
  = \mathbf{I}
    - \boldsymbol{\Sigma}_{\mathbf{e}}\,
      \mathbf{U}_{\!\mathbf{xe}}\,
      \operatorname{diag}\bigl(\boldsymbol{\lambda}_{\mathbf{xe}}^{-1/2}\bigr)\,
      \operatorname{diag}(\hat{\mathbf{q}})\,
      \mathbf{U}_{\!\mathbf{xe}}^{-1}.
      \label{eq:S_hat_colored_undirected}
\end{equation}

Accordingly, the expression for \(\hat{\mathbf{S}}\) in problem~\eqref{eq:undirected_structural_opt_2} is replaced by the one derived above.

Similarly, for a directed graph, we want to enforce  
\begin{equation}
\label{eq:dir_match}
\hat{\mathbf{H}}\,\boldsymbol{\Sigma}_{\mathbf{e}}\,\hat{\mathbf{H}}^{\top}
= \mathbf{C}_{\mathbf x}.
\end{equation}
To do this, we write the Cholesky decomposition  \(\mathbf{C}_{\mathbf x}=\mathbf{L}_{\mathbf x}\mathbf{L}_{\mathbf x}^{\top}\) and let 
\(\boldsymbol{\Sigma}_{\mathbf e}^{1/2}\) be the principal square root of the noise covariance matrix.  
Then \eqref{eq:dir_match} is actually equal to 
\begin{equation}
  \hat{\mathbf{H}}\;\boldsymbol{\Sigma}_{\mathbf e}^{1/2}
  \bigl(\hat{\mathbf{H}}\;\boldsymbol{\Sigma}_{\mathbf e}^{1/2}\bigr)^{\!\top}
  =\mathbf{L}_{\mathbf x}\mathbf{L}_{\mathbf x}^{\top}.
\end{equation}
Hence, there exists an orthogonal matrix \(\hat{\mathbf V}\in\mathcal O(N)\) such that
\begin{equation}
  \hat{\mathbf{H}}\;\boldsymbol{\Sigma}_{\mathbf e}^{1/2}\,\hat{\mathbf V}
  =\mathbf{L}_{\mathbf x}.
\end{equation}
Substituting \(\hat{\mathbf{H}}=(\mathbf{I}-\hat{\mathbf S})^{-1}\) and rearranging gives
\begin{equation}
  \hat{\mathbf S}
  =\mathbf I-
    \boldsymbol{\Sigma}_{\mathbf e}^{1/2}\,
    \hat{\mathbf V}\,
    \mathbf{L}_{\mathbf x}^{-1}.
\end{equation}
To incorporate this result into the optimization framework, simply replace the expression for \(\hat{\mathbf S}\) in problem~\eqref{eq:directed_structural_opt_2} with the formula derived above.

Both substitutions preserve the optimization problem structure; all algorithms remain applicable.
Hence, our method naturally extends to the colored-noise setting.

\subsection{Proof of Identifiability for Undirected Graph}
\label{appendix:proof_of_identifiability}

\begin{theorem}
Let the true GSO be $\mathbf S$ and  $\mathbf H = (\mathbf I-\mathbf S)^{-1}$. 
Let the EVD of $\mathbf{H}$ be $\mathbf H = \mathbf U \operatorname{diag}(\boldsymbol\lambda)\mathbf U^\top$. Further, consider the following two conditions:
\begin{enumerate}[label=(\roman*)]
    \item The matrix $\mathbf H$ has no two eigenvalues that are negatives of each other, i.e., $\nexists\, i\neq j:\ \lambda_i = -\lambda_j ;$
    \item The binary linear system
    \begin{equation}
    \label{eq:theorem_condition}
        (\mathbf U\circ\mathbf U)\,\operatorname{diag}(|\boldsymbol\lambda|^{-1})\,\hat{\mathbf q}
        = \mathbf 1,
        \qquad \hat{\mathbf q}\in\{\pm1\}^N,
    \end{equation}
    has a unique solution (since $\operatorname{sign}(\boldsymbol\lambda)$ is always a solution, the unique solution is equal to $\operatorname{sign}(\boldsymbol\lambda)$).
\end{enumerate}
Now, let us formulate the optimization problem
\begin{equation}
\label{eq:P1}
\begin{split}
&\argmin_{\hat{\mathbf q}\in\{\pm1\}^N}\|\mathbf W\hat{\mathbf q}-\mathbf 1\|_2^2,\\
&\text{where}\quad
\mathbf W   = (\mathbf U_{\mathbf x}\circ\mathbf U_{\mathbf x})\,
\operatorname{diag}(\boldsymbol\lambda_{\mathbf x}^{-1/2}) .
\end{split}
\end{equation}
By solving the optimization problem~\eqref{eq:P1}, we obtain a sign vector $\mathbf q^*$.
Then the resulting estimator $\mathbf S^*=\mathbf I-\mathbf U_{\mathbf x}\operatorname{diag}(\mathbf q^*)\,\operatorname{diag}(\boldsymbol\lambda_{\mathbf x}^{-1/2})\,\mathbf U_{\mathbf x}^\top$
satisfies $\mathbf S^* \to \mathbf S$ as $T\to\infty$.
\end{theorem}

\begin{proof}

At $T=\infty$, we have
\begin{equation}
\label{eq:Cx_undirected_decomp}
\mathbf C_{\mathbf x}
= \mathbf H^2
= \mathbf U \operatorname{diag}(\boldsymbol\lambda^2)\mathbf U^\top.
\end{equation}
On the other hand, we can write the EVD of $\mathbf C_{\mathbf x}$ as
\begin{equation}
\label{eq:Cx_EVD_undirected}
  \mathbf C_{\mathbf x}
  = \mathbf U_{\mathbf x}\operatorname{diag}(\boldsymbol\lambda_{\mathbf x})
    \mathbf U_{\mathbf x}^\top .
\end{equation}
Observing that both \eqref{eq:Cx_undirected_decomp} and
\eqref{eq:Cx_EVD_undirected} are EVD forms of the same covariance
matrix $\mathbf C_{\mathbf x}$, and using Assumption~(i), if
$\lambda_i^2=\lambda_j^2$ then necessarily $\lambda_i=\lambda_j$.
Hence every eigenspace of $\mathbf C_{\mathbf x}$ coincides with an
eigenspace of $\mathbf H$, and we can choose the EVD such that
\begin{equation}
\mathbf U_{\mathbf x}=\mathbf U, 
\qquad 
\boldsymbol\lambda_{\mathbf x}^{1/2} = |\boldsymbol\lambda|.
\end{equation}

Using $\mathrm{diag}\big(\mathbf A\,\mathrm{diag}(\mathbf v)\,\mathbf A^\top\big)
=(\mathbf A\circ\mathbf A)\mathbf v$, for any $\hat{\mathbf q}\in\{\pm1\}^N$ we have
\begin{equation}
\begin{aligned}
    \mathbf W\hat{\mathbf q}
&= (\mathbf U_{\mathbf x}\circ\mathbf U_{\mathbf x})\,\boldsymbol\lambda_{\mathbf x}^{-1/2}\hat{\mathbf q} \\
&= \mathrm{diag}\!\Big(
    \mathbf U_{\mathbf x}\,\mathrm{diag}\big(\boldsymbol\lambda_{\mathbf x}^{-1/2}\hat{\mathbf q}\big)\,
    \mathbf U_{\mathbf x}^\top
\Big) \\
&= \mathrm{diag}\!\Big(
    \mathbf U\,\mathrm{diag}\big(|\boldsymbol\lambda|^{-1}\hat{\mathbf q}\big)\,
    \mathbf U^\top
\Big) \\
&= (\mathbf U\circ\mathbf U)\,\operatorname{diag}(|\boldsymbol\lambda|^{-1})\,\hat{\mathbf q}.
\end{aligned}
\end{equation}
By Assumption (ii), the equation
$(\mathbf U\circ\mathbf U)\,\operatorname{diag}(|\boldsymbol\lambda|^{-1})\,\hat{\mathbf q}
=\mathbf 1$
has the unique binary solution $\operatorname{sign}(\boldsymbol\lambda)$.
Hence the minimizer of \eqref{eq:P1} satisfies
$\mathbf q^*=\operatorname{sign}(\boldsymbol\lambda)$ for $T=\infty$.

Finally, the estimate $\mathbf S^*$ then converges to the ground truth $\mathbf S$ as
\begin{equation}
\begin{aligned}
\mathbf S^*
&= \mathbf I - \mathbf U_{\mathbf x}\operatorname{diag}(\mathbf q^*)\,
   \operatorname{diag}(\boldsymbol\lambda_{\mathbf x}^{-1/2})\,\mathbf U_{\mathbf x}^\top \\
&= \mathbf I - \mathbf U\,\operatorname{diag}(\operatorname{sign}(\boldsymbol\lambda))\,
   \operatorname{diag}(|\boldsymbol\lambda|^{-1})\,\mathbf U^\top \\
& = \mathbf I - \mathbf U\,\operatorname{diag}(\boldsymbol\lambda^{-1})\,\mathbf U^\top = \mathbf I - \mathbf H^{-1} = \mathbf S .
\end{aligned}
\end{equation}

\end{proof}

\subsection{Proof of Identifiability for Directed Graph}
\label{appendix:proof_of_identifiability_2}

Since $\mathbf H$ is not necessarily symmetric in the directed case, we work
with its SVD, i.e., $\mathbf H = \mathbf U \operatorname{diag}(\boldsymbol\lambda)\mathbf V^\top$.
At $T=\infty$, the covariance matrix then satisfies
\begin{equation}
\label{eq:Cx_directed_decomp}
\mathbf C_{\mathbf x}
= \mathbf H\mathbf H^\top
= \mathbf U \operatorname{diag}(\boldsymbol\lambda^2)\mathbf U^\top,
\end{equation}
which is similar to the expression for the undirected case in~\eqref{eq:Cx_undirected_decomp}.

On the other hand, we can write the EVD of $\mathbf C_{\mathbf x}$ as
\begin{equation}
\label{eq:Cx_EVD_directed}
  \mathbf C_{\mathbf x}
  = \mathbf U_{\mathbf x}\operatorname{diag}(\boldsymbol\lambda_{\mathbf x})
    \mathbf U_{\mathbf x}^\top .
\end{equation}

Observing that both \eqref{eq:Cx_directed_decomp} and
\eqref{eq:Cx_EVD_directed} are EVD forms of the same covariance
matrix $\mathbf C_{\mathbf x}$, their eigenvalues and eigenvectors must match. Furthermore,  since the singular
values in $\boldsymbol\lambda$ are nonnegative, we have 
\begin{equation}
 \mathbf U =  \mathbf U_{\mathbf x},
  \qquad
  \boldsymbol\lambda
  = \boldsymbol\lambda_{\mathbf x}^{1/2}.
\end{equation}

\subsection{Derivative in Euclidean Space}
\label{appendix:derivative}

In the directed case, problem~\eqref{eq:directed_structural_opt_2} reduces to an orthogonality-constrained minimization over $\hat{\mathbf V}\in\mathcal{O}(N)$ with the objective function $\mathcal{J}(\hat{\mathbf V})$ defined in~\eqref{eq:obj_J}
\begin{equation}
    \begin{aligned}
\label{eq:directed_objective}
\mathcal{J}(\hat{\mathbf V})
&= \Big\|\mathrm{diag}\!\Bigl(
\hat{\mathbf V}\,\mathrm{diag}(\boldsymbol{\lambda}_{\mathbf x}^{-1/2})
\,\mathbf U_{\mathbf x}^{\top}-\mathbf I\Bigr)\Big\|_F^2
\\
&\quad
+ \alpha\Big\|\hat{\mathbf V}\,\mathrm{diag}(\boldsymbol{\lambda}_{\mathbf x}^{-1/2})
\,\mathbf U_{\mathbf x}^{\top}-\mathbf I\Big\|_1 .
\end{aligned}
\end{equation}

To apply Riemannian methods, we need to compute the Euclidean gradient of $\mathcal{J}$ in $\mathbb{R}^{N\times N}$.  
 For notational convenience, let
$
\hat{\mathbf S}= \mathbf I-\hat{\mathbf V}\,\mathrm{diag}(\boldsymbol\lambda_{\mathbf x}^{-1/2})\,\mathbf U_{\mathbf x}^{\top},
$
then
$\mathrm d\hat{\mathbf S}=-\,\mathrm d\hat{\mathbf V}\,
(\mathrm{diag}(\boldsymbol\lambda_{\mathbf x}^{-1/2})\,\mathbf U_{\mathbf x}^{\top})$.

Using the Frobenius product identity $\langle A,BC\rangle=\langle AC^{\top},B\rangle$, for any scalar function $f$, we obtain
\[
\begin{aligned}
\mathrm df
&= \big\langle \nabla_{\hat{\mathbf S}} f(\hat{\mathbf S}),\,\mathrm d\hat{\mathbf S}\big\rangle\\
&= \Big\langle \nabla_{\hat{\mathbf S}} f(\hat{\mathbf S}),\, -\,\mathrm d\hat{\mathbf V}\,
\big(\mathrm{diag}(\boldsymbol\lambda_{\mathbf x}^{-1/2})\,\mathbf U_{\mathbf x}^{\top}\big)\Big\rangle
\\
&= -\,\Big\langle \nabla_{\hat{\mathbf S}} f(\hat{\mathbf S})\,
\big(\mathrm{diag}(\boldsymbol\lambda_{\mathbf x}^{-1/2})\,\mathbf U_{\mathbf x}^{\top}\big)^{\!\top},\,
\mathrm d\hat{\mathbf V}\Big\rangle,
\end{aligned}
\]
hence $\nabla_{\hat{\mathbf V}}\,f(\hat{\mathbf S})
= -\,\big(\nabla_{\hat{\mathbf S}} f(\hat{\mathbf S})\big)\,
\mathbf U_{\mathbf x}\,\mathrm{diag}(\boldsymbol\lambda_{\mathbf x}^{-1/2}).$

Therefore, for the diagonal penalty, using $\nabla_{\mathbf X}\|\mathrm{diag}(\mathbf X)\|_F^2=2\,\mathrm{Diag}(\mathbf X)$, we get
\begin{equation}
\begin{aligned}
& \quad\frac{\partial}{\partial \hat{\mathbf V}}
\Big\|\mathrm{diag}\!\big(
\hat{\mathbf V}\,\mathrm{diag}(\boldsymbol{\lambda}_{\mathbf x}^{-1/2})
\,\mathbf U_{\mathbf x}^{\top}-\mathbf I\big)\Big\|_F^2
\\&= \frac{\partial}{\partial \hat{\mathbf V}}
\Big\|\mathrm{diag}\!\big(-\hat{\mathbf S}\big)\Big\|_F^2
\\
&= -\,2\,\mathrm{Diag}(\hat{\mathbf S})\,
\mathbf U_{\mathbf x}\,
\mathrm{diag}(\boldsymbol{\lambda}_{\mathbf x}^{-1/2})
\\
&= -\,2\,\mathrm{Diag}\!\Big(
\mathbf I-
\hat{\mathbf V}\,\mathrm{diag}(\boldsymbol{\lambda}_{\mathbf x}^{-1/2})
\,\mathbf U_{\mathbf x}^{\top}\Big)\,
\mathbf U_{\mathbf x}\,
\mathrm{diag}(\boldsymbol{\lambda}_{\mathbf x}^{-1/2}) .
\end{aligned}
\end{equation}

For the sparsity penalty, using the elementwise subgradient $\partial\|\mathbf X\|_1=\mathrm{sign}(\mathbf X)$ (with $\mathrm{sign}(0)\in[-1,1]$), we obtain
\begin{equation}
\begin{aligned}
&\quad \frac{\partial}{\partial \hat{\mathbf V}}
\Big\|\hat{\mathbf V}\,
\mathrm{diag}(\boldsymbol{\lambda}_{\mathbf x}^{-1/2})
\,\mathbf U_{\mathbf x}^{\top}-\mathbf I\Big\|_1
\\&= \frac{\partial}{\partial \hat{\mathbf V}}\big\|-\hat{\mathbf S}\big\|_1
\\
&\in \ \mathrm{sign}\!\big(\hat{\mathbf V}\,
\mathrm{diag}(\boldsymbol{\lambda}_{\mathbf x}^{-1/2})
\,\mathbf U_{\mathbf x}^{\top}-\mathbf I\big)\,
\mathbf U_{\mathbf x}\,\mathrm{diag}(\boldsymbol{\lambda}_{\mathbf x}^{-1/2})
\\
&= -\,\mathrm{sign}\!\Big(
\mathbf I-\hat{\mathbf V}\,
\mathrm{diag}(\boldsymbol{\lambda}_{\mathbf x}^{-1/2})
\,\mathbf U_{\mathbf x}^{\top}\Big)\,
\mathbf U_{\mathbf x}\,
\mathrm{diag}(\boldsymbol{\lambda}_{\mathbf x}^{-1/2}),
\end{aligned}
\end{equation}

Since $\operatorname{sign}(\cdot)$ is discontinuous and often induces pronounced oscillations in the objective, in practice we optimize with a Huber-type smooth surrogate for $\|\cdot\|_1$ ~\cite{huber1992robust} and revert to $\operatorname{sign}(\cdot)$ only once the loss has stabilized.
Concretely, if the objective shows no improvement for 30 consecutive inner iterations $l$ in Algorithm~\ref{alg:basinhop_riemann_multi_candidates}, we consider it stabilized and replace the Huber-like gradient by the discrete 
$\operatorname{sign}(\cdot)$ function.

\subsection{Candidate Set Update}
\label{appendix:candidate_update}

To maintain diversity, candidates are admitted in ascending order of 
objective value, but only if they are sufficiently different from those 
already selected. Specifically, let $\delta>0$ be a distance threshold. 
We traverse the refined set $\mathcal{R}$ in order of increasing 
$\mathcal{J}(\hat{\mathbf V})$ and add a matrix $\hat{\mathbf V}$ to the 
candidate set $\mathcal{C}$ only if
\begin{equation}
    \min_{\hat{\mathbf W}\in\mathcal{C}}
    \|\hat{\mathbf V}-\hat{\mathbf W}\|_F^2 > \delta .
\end{equation}
This process continues until either $L'$ candidates are collected or 
$\mathcal{R}$ is exhausted. The resulting $\mathcal{C}$ therefore contains 
the $L'$ best solutions while avoiding near-duplicate matrices.

In practice, we adapt $\delta$ depending on the optimization dynamics 
discussed in Appendix~\ref{appendix:derivative}: when the loss has not yet 
stabilized we set $\delta=0.5$, whereas after stabilization we use 
$\delta=0.01$.

\end{document}